\documentclass[conference]{IEEEtran}
\IEEEoverridecommandlockouts

\usepackage{amsmath,amssymb,amsfonts,amsthm}
\usepackage{algorithm}
\usepackage{algorithmic}
\usepackage{graphicx}
\usepackage{textcomp}
\usepackage{empheq}
\usepackage{xcolor}

\usepackage{tikz}
\definecolor{Green}{rgb}{0.0, 0.5, 0.0}
\usepackage{diagbox}
\usepackage{subcaption}

\def\BibTeX{{\rm B\kern-.05em{\sc i\kern-.025em b}\kern-.08em
T\kern-.1667em\lower.7ex\hbox{E}\kern-.125emX}}

\newtheorem{thm}{Theorem}
\newtheorem{defn}{Definition}
\newtheorem{rem}{Remark}
\newtheorem{lem}{Lemma}
\newtheorem{prop}{Proposition}
\newtheorem{cor}{Corollary}
\newtheorem*{pf}{Proof}
\newtheorem{assum}{Assumption}

\newif\ifitsdraft
\def\itsdraft{\global\itsdrafttrue}

\itsdraft

\begin{document}

\title{How to discretize continuous state-action spaces in Q-learning: A symbolic control approach
}


\author{\IEEEauthorblockN{Sadek Belamfedel Alaoui and Adnane Saoud}
\IEEEauthorblockA{\textit{College of Computing, University Mohammed VI Polytechnic, Ben Guerir 43150, Morocco} \\
sadek.belamfedel, adnane.saoud (at) um6p.ma}
}

\maketitle

\begin{abstract}
Q-learning is widely recognized as an effective approach for synthesizing controllers to achieve specific goals. However, handling challenges posed by continuous state-action spaces remains an ongoing research focus. This paper presents a systematic analysis that highlights a major drawback in space discretization methods. To address this challenge, the paper proposes a symbolic model that represents behavioral relations, such as alternating simulation from abstraction to the controlled system. This relation allows for seamless application of the synthesized controller based on abstraction to the original system. Introducing a novel Q-learning technique for symbolic models, the algorithm yields two Q-tables encoding optimal policies. Theoretical analysis demonstrates that these Q-tables serve as both upper and lower bounds on the Q-values of the original system with continuous spaces. Additionally, the paper explores the correlation between the parameters of the space abstraction and the loss in Q-values.  The resulting algorithm facilitates achieving optimality within an arbitrary accuracy, providing control over the trade-off between accuracy and computational complexity. The obtained results provide valuable insights for selecting appropriate learning parameters and refining the controller. \textcolor{black}{The engineering relevance of the proposed Q-learning based symbolic model is illustrated through two case studies.}
\end{abstract}


\section{Introduction}


Q-learning, introduced by Watkins and Dayan \cite{watkins1989learning}, has become a prominent research area in reinforcement learning, garnering significant attention for its ability to autonomously learn optimal policies. Its successful applications span diverse domains, including games, robotics, and control systems, as evidenced by notable works such as \cite{mnih2015human}. The allure of Q-learning lies in its simplicity and effectiveness in tackling problems with discrete state and action spaces. However, real-world scenarios often present challenges that traditional Q-learning approaches struggle to address, particularly when dealing with continuous state-action spaces \cite{mnih2013playing, gu2016continuous, van2020q}.

Continuous state-action spaces are characterized by an infinite number of possible state-action pairs, rendering the explicit tabular representation of states impractical \cite{sutton2018reinforcement}. Consequently, there has been a growing interest in extending Q-learning techniques to handle continuous state-action spaces, as seen in works such as \cite{mnih2013playing, gu2016continuous, van2020q, seyde2022solving}. Among these techniques, space discretization-based methods have been explored in \cite{powell2007approximate, lampton2009multiresolution, lampton2010multi, hoerger2022adaptive}.

Space discretization methods partition continuous state and action spaces into discrete subsets called cells, viewing system trajectories as trajectories of these discrete cells. However, the straightforward approach of space discretization, as described in earlier works \cite{powell2007approximate, lampton2009multiresolution, lampton2010multi, hoerger2022adaptive}, can lead to an undesirable mismatch in the reachable set of points within an action. The state space is divided into a uniform grid, with each point represented by the center of the corresponding cell. The successor state is computed based on the system's difference equation, but only for the center of the current cell. Consequently, the trajectory of a point within the starting cell may not end up in the same cell as its center. Figure \ref{fig:uniformdiscretisation} visually demonstrates this reachable-induced mismatch.
\begin{figure}
	\centering
	\includegraphics[width=0.7\linewidth]{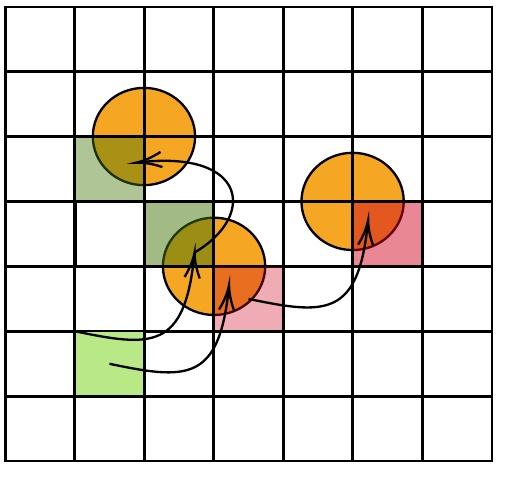}
	\caption{Mismatch between the actual system trajectory and the trajectory obtained by uniform discretization under the same policy $\pi$. The green cells represent the actual transitions whereas the red cells represent the transition captured by the uniform discretisation.}
	\label{fig:uniformdiscretisation}
\end{figure} 

To address the Q-learning problem for continuous state action spaces, this paper presents a novel approach using symbolic abstraction to mitigate underestimation of the reachable set by intersecting the it with multiple discrete cells. This over-approximation introduces conservatism, rendering the discrete system non-deterministic \cite{tabuada2009verification}. Building upon the Q-learning based approach for non-deterministic transition systems proposed in \cite{borri2023reinforcement}, this paper introduces a new Q-learning based discretized reward. As a result of this extension, two Q-tables are generated: the minimal and maximal Q-values, effectively bounding the Q-values of the continuous state-action spaces. The theoretical development establishes conditions under which the minimal and maximal Q-values converge to an optimal Q-values. Additionally, a direct relation between the tightness of these Q-values and the abstraction parameters is established. The analysis demonstrates that under specific learning configurations, the bounds become tighter, leading to a convergence of the extracted policy towards the optimal policy of the continuous state-action space. \textcolor{black}{The effectiveness of the proposed algorithm is evaluated through its application to the mountain car control problem \cite{Moore90efficientmemory} and to the Van Der Pol oscillator.} The evaluation reveals that either the maximally or minimally computed Q-values can effectively control the system. Moreover, the proposed algorithm can achieve any desired precision on the Q-values for continuous state-action spaces, resulting in the convergence of the extracted policy towards the optimal policy of the continuous spaces. Furthermore, the evaluation shows that reducing the distance between quantizer levels leads to tighter bounds and improves the accuracy of the approximation of the optimal Q-values in terms of similarity between the maximally and minimally computed policies.

\section{Related works}

Space discretization in Q-learning is a thriving area of research, driven by the need to adapt Q-learning techniques to handle continuous state-action spaces. To discretize a Q-learning problem, it is necessary to analyze how to construct subgoals for the agent. In this context, \cite{bentivegna2003learning} presents an algorithm that combines Q-learning with a locally weighted learning method to select behavioral primitives and generate subgoals for the agent. Similarly, \cite{csimcsek2005identifying} focuses on identifying subgoals by partitioning local state transition graphs. Quad-Q learning, as developed by \cite{Clausen2000}, contributes to the theory of Q-learning by extending it to a quad-based framework, providing an innovative solution applicable to learning how to partition large intractable problem domains into smaller problems that can be solved independently. Hierarchical Reinforcement Learning (HRL) methods have also been explored to address space discretization challenges. MAXQ, introduced by \cite{dietterich1998maxq}, decomposes a learning problem into a hierarchy of subtasks, each of which is learned using Q-learning. Grid-based discretization methods have been extensively investigated. For example, \cite{powell2007approximate} introduces an approximate dynamic programming approach, utilizing grid-based discretization to address high-dimensional continuous control problems. Additionally, \cite{lampton2009multiresolution} explore a multiresolution approach, where the grid resolution varies based on the attainable subgoals. Building on this approach,  \cite{lampton2010multi} extends the multiresolution technique with pseudo-randomized discretization to enhance its adaptability. The works \cite{sinclair2020adaptive} and \cite{sinclair2022adaptive} revolve around the concept of adaptive discretization, primarily achieved through non-uniform partitioning of the space. This approach involves continuously refining the partition based on the density of the observed samples, offering valuable insights into data-driven approaches for achieving more efficient and effective learning. Moreover, \cite{hoerger2022adaptive} propose an adaptive space partitioning technique that automatically adjusts the granularity of the discretization based on the system's exploration. \textcolor{black}{The work in \cite{delimpaltadakis2023interval} establishes an algorithm to solve value iteration over continuous actions interval Markov Decision Processes that uses maximization problems instead of maximisation - minimisation for efficiency.} In the context of controlling non-deterministic finite transition systems, \cite{borri2023reinforcement} introduced a Q-learning algorithm tailored to address the challenges posed by such systems. Non-deterministic finite transition systems are discrete systems whose transition probabilities are not explicitly defined. However, it is essential to note that the work presented here and the study conducted by \cite{borri2023reinforcement} are distinct in their problem statements and objectives.


\section{Preliminaries and problem statement}

\subsection{Notation}
The symbols $\mathbb{N}$, $\mathbb{N}_{> 0}$, $\mathbb{R}$, and $\mathbb{R}_{\geqslant 0}$ represent the sets of natural numbers, strictly positive natural numbers, real numbers, and positive real numbers, respectively. When referring to a set $X$, the notation $2^{X}$ denotes the set of all subsets of $X$. Given any $a \in \mathbb{R},|a|$ denotes the absolute value of $a$. Given any $u=\left(u_{1}, \ldots, u_{n}\right) \in \mathbb{R}^{n}$, the infinity norm of $u$ is defined by $\|u\|=\max\limits_{1 \leqslant i \leqslant n}\left|u_{i}\right|$. Given sets $X$ and $Y$, we denote by $f: X \rightarrow Y$ an ordinary map of $X$ into $Y$. We denote the closed, open, and half-open intervals in $\mathbb{R}$ by $[a, b],(a, b),[a, b)$, and $(a, b]$, respectively. For any set $S \subseteq \mathbb{R}^{n}$ of the form $S=\bigcup_{j=1}^{M} S_{j}$ for some $M \in \mathbb{N}$, where $S_{j}=\prod_{i=1}^{n}\left[c_{i}^{j}, d_{i}^{j}\right] \subseteq \mathbb{R}^{n}$ with $c_{i}^{j}<d_{i}^{j}$, and non-negative constant $\eta \leqslant \tilde{\eta}$, where $\tilde{\eta}=\min\limits_{j=1, \ldots, M} \eta_{S_{j}}$ and $\eta_{S_{j}}=\min \left\{\left|d_{1}^{j}-c_{1}^{j}\right|, \ldots,\left|d_{n}^{j}-c_{n}^{j}\right|\right\}$, we define $[S]_{\eta}=\left\{a \in S \mid a_{i}=k_{i} \eta, k_{i} \in \mathbb{Z}, i=1, \ldots, n\right\}$ if $\eta \neq 0$, and $[S]_{\eta}=S$ if $\eta=0$. The set $[S]_{\eta}$ will be used as a finite approximation of the set $S$ with precision $\eta \neq 0$. Note that $[S]_{\eta} \neq \emptyset$ for any $\eta \leqslant \tilde{\eta}$. $\mathbb{B}= \left\lbrace \textbf{x}\in \mathbb{R}^{n} | \;\; \|\textbf{x}\|_{\infty}\leqslant 1 \right\rbrace $ denotes the unit ball.


\subsection{System dynamics}
Consider a control problem in which the system evolves based on a discrete-time dynamical equation defined as follows:
\begin{align}\label{Eq1}
(\Sigma): \xi_{k+1} = f(\xi_{k}, v_{k}),
\end{align}
the environment state is denoted as $\xi_{k}$ and belongs to the continuous space $\mathcal{S}\subseteq \mathbb{R}^{n}$, while the control input is represented by $v_{k}: \mathcal{A}(\xi_{k}) \rightarrow \mathcal{S} $, with $\mathcal{A}(\xi_{k})$ the set of admissible actions that belongs to the set of all actions $\mathcal{A}$, i.e. $\mathcal{A}(\xi_{k})\subseteq \mathcal{A} \subseteq \mathbb{R}^{m}$. The actions are decisions made by a controller at each time step $k\in \mathbb{N}$. The map $f: \mathcal{S}\times \mathcal{A} \to \mathcal{S}$ is assumed to be known and generally nonlinear. Let $\pi = (v_{0}, v_{1}, \ldots) \in \Pi \subseteq 2^{\mathcal{A}}$, with $v_{k} \in \mathcal{A}(\xi_{k})$, represents a potentially infinite sequence of actions. The notation $\phi(k, \xi_0, \pi)$ is used to denote the state reached at time $k$ starting from the initial state $\xi_0$ under the sequence of actions $\pi$. 

In the subsequent analysis, we make the assumption that the map $f$ fulfils the following Lipschitz assumption.
\begin{assum}\label{Assum1}
There exist Lipschitz constants $L_{f\xi}$, $L_{fv}$ such that for any $\xi, \overline{\xi} \in \mathcal{S}$ and $v \in \mathcal{A}(\xi)$, $\overline{v}\in \mathcal{A}(\overline{\xi})$, the following inequality holds:
\begin{align}\label{EqAssum1}
\left\|f(\xi, v) - f(\overline{\xi}, \overline{v})\right\| \leqslant  L_{f\xi}\|\xi - \overline{\xi}\| + L_{fv}\|v -\overline{v}\|
\end{align}
where $ \mathcal{A}(\xi)$ and $ \mathcal{A}(\overline{\xi})$ denote the sets of admissible control inputs at state $ \xi$, $\overline{\xi}$ respectively. 
\end{assum}

The following assumption is used through the paper.
\begin{assum}\label{Compact}
The state and action space $\mathcal{S}$ and $\mathcal{A}$ are compact.
\end{assum}

The following auxiliary lemma will be used to prove the main results of the paper.

\begin{lem}\label{Lemma1}
Under Assumption \ref{Compact}, there exists a positive constant $ L_{\mathcal A}$ such that for every distinct $\xi, \overline{\xi} \in \mathcal{S} $,  $v \in \mathcal{\mathcal{A}(\xi)}$ and $\overline{v} \in \mathcal{A}(\overline{\xi})$ we have: 
\begin{align}\label{EqLemma1}
\| v-\overline{v} \| \leqslant   L_{\mathcal A} \|\xi-\overline{\xi}\|,
\end{align}
\end{lem}

\ifitsdraft
\begin{pf}
Suppose that $ \mathcal{A}$ is a compact set, then there exist a constant $C$ such that $\| v-\overline{v} \| \leqslant  C $. Let's prove that \eqref{EqLemma1} holds for any distinct $\xi, \overline{\xi} \in \mathcal{S}$ and $v \in \mathcal{A}(\xi)$, $\overline{v} \in \mathcal{A}(\overline{\xi})$. Inequality \eqref{EqLemma1} can be written as $\| v-\overline{v} \| \leqslant  C  \|\xi-\overline{\xi}\|  \frac{1}{\|\xi-\overline{\xi}\|}$.
Let $ M= \max\limits_{\xi,\bar{\xi} \in \mathcal{S}} \frac{1}{\|\xi-\overline{\xi}\|} $ the existence of $M \geq 0$ is guaranteed since $\xi$ and $\bar{\xi}$ are distinct and the state space $\mathcal{S} $ is compact.
Thus, $\| v-\overline{v} \| \leqslant  C  M \|\xi-\overline{\xi}\| $ implies that $\| v-\overline{v} \| \leqslant   L_{\mathcal A}  \|\xi-\overline{\xi}\|$, where $ L_{\mathcal A} = C M$ is a positive constant. This completes the proof. \hfill $\square$
\end{pf}

\fi

\subsection{Why earlier discretisation methods are not performing for Q-learning?}

In the context of the control system $\Sigma$, the Q-learning control scheme \cite{watkins1989learning} is employed to learn the Q-values through a sequence of observations, actions, and rewards. To guide the algorithm towards achieving the desired goal, a reward map $g: \mathcal{S}\times \mathcal{A} \to \mathbb{R}$ is associated with each state-action pair. Indeed, at each time step $k \in \mathbb{N}$, the current state $\xi_{k}$ is observed and a decision $v_{k}$ is selected from the admissible set of control inputs $\mathcal{A}\left(\xi_{k}\right)$. After $v_{k}$ is performed, the system goes to a next state $\xi_{k+1}=\xi^{\prime}$. Associated with this state transition, an immediate reward $g(\xi_{k},v_{k})$ is gained. The object of the controller is to find an optimal policy that maximizes $\sum_{i=0}^{\infty} \gamma^{i} g(\xi_{k+i},v_{k+i})$, where $\gamma \in \left(0,1\right)$ is the discount factor. Taken arbitrary policy $\pi$, the $Q$-values are defined by:
\begin{align}\label{update}
q_{\pi}(\xi,a)=g(\xi,a) + \gamma  \max\limits_{b^{\prime}\in\mathcal{A}(\xi^{\prime})}q_{\pi}(\xi^{\prime},b^{\prime}).
\end{align}
The objective in $Q$-Learning is to estimate the $Q$-values for an optimal policy when the reward function is known a priori. The optimal Q-values, denoted $q_{\pi}^{\ast}(\xi, v)$ are generally learned by Algorithm \ref{alg:q_learning}.
\begin{algorithm}
\caption{Q-Learning Algorithm}
\label{alg:q_learning}
\begin{algorithmic}[1]
\STATE Initialize the Q-table, $\alpha$ learning rate and $\gamma$ the discount factor.
\STATE Define the state and action spaces
\FOR{each episode}
\STATE Initialize the current state
\WHILE{episode not finished}
\STATE Choose an action based on the current state and exploration-exploitation strategy
\STATE Apply the action to $\Sigma$ and observe the reward $g(\xi,v)$ and next state $\xi'$
\STATE Update the Q-value for the current state-action pair using the Q-learning update equation:
\begin{equation*}
\resizebox{0.85\hsize}{!}{$Q(\xi, v) \leftarrow Q(\xi, v) + \alpha \left(g(\xi,v) + \gamma \max\limits_{v'\in \mathcal{A}(\xi')} Q(\xi', v')\right)$}
\end{equation*}
\STATE Set the current state to the next state
\ENDWHILE
\ENDFOR
\STATE Extract the learned policy from the Q-table
\end{algorithmic}
\end{algorithm}

The next result from \cite{watkins1989learning} shows under which conditions Algorithm \ref{alg:q_learning} converges.
\begin{thm} \label{watkinsTHM}
Given bounded $\left|g\right|\leqslant  \Re \in \mathbb{R}_{\geqslant 0}$, learning rates $0\leqslant \alpha<1$ and 
$$\sum_{k=1}^{\infty} \alpha_{k}(\xi, a)=\infty, \sum_{k=1}^{\infty}\left[\alpha_{k}(\xi, a)\right]^{2}<\infty,$$
then $q_{\pi}^{(k)}(\xi,a)\rightarrow q^{\star}$ as $k \rightarrow \infty $ $\forall \xi, a$ with probability $1$.
\end{thm}

The problem of applying Algorithm \ref{alg:q_learning} to the system $\Sigma$ relies on the fact that $\xi $ and $ v $ are both defined in a continuous space making updating the Q-values impossible over an infinite number of state action pairs \cite{powell2007approximate,sutton2018reinforcement,seyde2022solving}. To overcome this problem, discretization methods have been proposed to effectively handle infinite state and action spaces. One such method is uniform discretization \cite{powell2007approximate,lampton2009multiresolution,buadicua2022experiments}, which divides the state and action spaces into evenly sized intervals. Another approach is Voronoi discretization \cite{hoerger2022adaptive}, which partition the spaces based on the proximity of the states and actions to specific points. Algorithm \ref{alg:q_learning_discretisation} provides a detailed description of how Q-learning with uniform discretization operates.

It is worth noting that such discretisation methods (uniform and Voronoi) can result in a discretization-induced reachable mismatch, see Figure \ref{fig:uniformdiscretisation}. In these approaches, the state space is divided into a uniform grid, and each point is represented by the centre of the cell to which it belongs. The successor state is computed based on the system's difference equation \eqref{Eq1}, but only for the centre of the current cell. As a consequence, the trajectory of a point within the starting cell may not end up in the same cell as its centre. This mismatch arises because the discretisation method does not take into account the trajectories that cross cell boundaries, causing the actual trajectory to deviate from the idealized trajectory represented by the cell centres. Relying solely on the points $x_{i} \in G $, see Algorithm \ref{alg:q_learning_discretisation} to compute successor states can lead to inaccuracies in the learned Q-values and the resulting policy.

\begin{algorithm}
\caption{Q-Learning with uniform discretisation}
\label{alg:q_learning_discretisation}
\begin{algorithmic}[0]
\STATE \textbf{Input:}  State set: $\mathcal{S}$, Action set: $\mathcal{A}$, Number of partitions for state set: $n$, Number of partitions for action set: $p$
\STATE \textbf{Output:}  Approximated Q-values for the discrete state-action pairs
\STATE \textbf{Partition the state set $\mathcal{S}$ into $n$ disjoint subsets:} $\mathcal{S}_{1}, \mathcal{S}_{2}, ..., \mathcal{S}_{n}$. Ensure that each subset covers the entire state space, i.e., $\mathcal{S} = \cup_{i=1}^{n} \mathcal{S}_i$. 
\STATE \textbf{Select a representative point} $x_i$ \textbf{from each subset} $\mathcal{S}_i$ \textbf{to form a finite aggregated state set} $G = \{x_1, x_2, ..., x_n\}$. Generally, $x_{i}$ is the centres of the cell $i$.  
\STATE \textbf{Partition the action set $\mathcal{A}$ into $p$ disjoint subsets:} $\mathcal{A}_{1}, \mathcal{A}_{2}, \dots, \mathcal{A}_{p}$. Ensure that each subset covers the entire action space, i.e., $\mathcal{A} = \cup_{j=1}^{p} \mathcal{A}_j$.
\STATE \textbf{Select a representative action} $a_j$ \textbf{from each subset} $\mathcal{A}_j$ \textbf{to form a finite aggregated decision set} $H = \{a_1, a_2, ..., a_p\}$.
\STATE \textbf{Initialize the Q-values} for each discrete state-action pair in $G \times H$ as $Q(x, u) = 0$, for all $x \in G$ and $u \in H$.
\FOR{each episode}
\STATE Initialize state $x$
\WHILE{episode not finished}
\STATE Choose action $a \in H $ using an exploration or exploitation strategy (e.g., epsilon-greedy)
\STATE Execute action $a$ and observe next state $\xi'$ and reward $g(x,a)$
\STATE Associate $\xi' \in \mathcal{S}$  with the closest $x' \in G$
\STATE Update the Q-values using the Q-learning update rule:\\
$Q(x, a) \leftarrow Q(x, a) + \alpha(g + \gamma\max\limits_{a'\in H} Q(x', a'))$
\STATE Update current state $x \in G $ by $ x' \in G$
\ENDWHILE
\ENDFOR
\end{algorithmic}
\end{algorithm}

\subsection{Problem Statement}

This paper addresses several key challenges in controlling system \eqref{Eq1} with continuous state and action spaces using Q-learning. The first challenge is to identify a suitable discretization method that effectively captures the essential dynamics of the system. Once a suitable discretization method is identified, the next challenge is to learn the optimal Q-values of the discrete system while evaluating the conservatism introduced by the discretization process. By analysing the parameters of the discretization, the paper aims to gain insight into establishing an arbitrary precision $\varepsilon $ between the resulting discrete Q-values and the Q-values of the original continuous system \eqref{Eq1}. To the best of our knowledge, these research questions represent novel and unexplored points for Q-learning.

\subsection{Solution strategy}

To address the research questions at hand, this paper proposes an abstraction-based Q-learning approach. This approach involves three steps, abstraction, Q-learning and refinement, see Figure \ref{PrincipleAbstraction}. In the "Abstraction", the state-input spaces  is partitioned into intervals called cells. By analysing the reachable set of a group of points belonging to a particular cell, the successor states are identified by the discrete states that intersect with the reachable set. This over-approximation transforms the original deterministic system described by \eqref{Eq1} into a nondeterministic one. Since state transitions are nondeterministic and not probabilstic, the traditional Q-learning technique as described in \cite{jiang1999convergence} cannot be applied. The second step introduces a  novel Q-learning algorithm with two Q-tables that captures the worst-case and best-case scenarios of the control system under a given policy. The convergence and uniqueness of these Q-tables are analyzed to ensure the effectiveness of the proposed approach. Finally, the refinement step permits to extract the optimal policy from the obtained optimal Q-values and apply one of the resultant optimal policies to system \eqref{Eq1}. 
\begin{figure}
\includegraphics[width=0.9\linewidth]{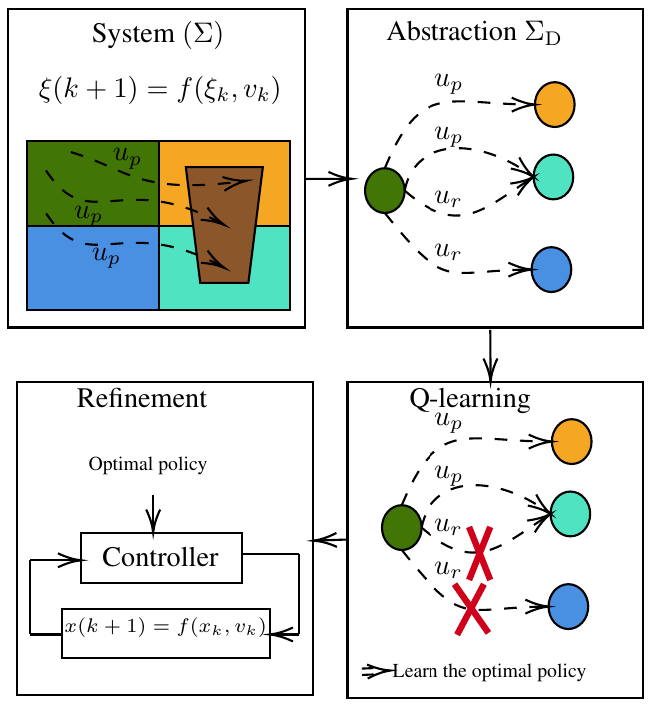}
\caption{Principle of Abstraction based Q-learning.}
\label{PrincipleAbstraction}
\end{figure}

\section{Symbolic models}

Define the notion of \textit{transition systems} adopted from \cite{tabuada2009verification}, which allows to represent the discrete time system \eqref{Eq1} and its abstraction in a unified way.
\begin{defn}
A transition system $\Sigma$ is a tuple $\Sigma =
(\mathcal{S}, \mathcal{A}, F, g, g)$, where $\mathcal{S}$ is a set of states, $\mathcal{A}$ is a set of
inputs, $F : \mathcal{S} \times \mathcal{A} \implies \mathcal{S}$ is a transition relation, $g: \; \mathcal{S} \times \mathcal{A} \rightarrow \mathbb{R} $ a reward function.
\end{defn}
For the  transition system $\Sigma$, the set of enabled input for a state $\xi \in \mathcal{S}$ is formally defined by, 
\begin{align}\label{key7}
v \in \mathcal{A}(\xi) \Longleftrightarrow f(\xi, v) \subseteq \mathcal{S},
\end{align}
and,
\begin{align}\label{key8}
\forall v \in \mathcal{A}(\xi), F(\xi, v)=f(\xi, v).
\end{align}

For transition systems, the notion of Alternating Simulation \cite{tabuada2009verification} represents a formal relationship between the behaviours of two systems. This relationship enables the refinement of a controller that has been synthesized for one system, facilitating effective control over the other system while preserving correctness guarantees.

\begin{defn}
Let $\Sigma_{1} = (\mathcal{S}_{1}, \mathcal{A}_{1}, \Delta_{1}) $ and $\Sigma_{2} = (\mathcal{S}_{2}, \mathcal{A}_{2}, \Delta_{2})$ be metric systems with $Y_1=Y_2$. A relation $\mathcal{R} \subseteq \Sigma_1 \times \Sigma_2 $ is an alternating simulation relation from $\Sigma_{1}$ to $\Sigma_{2}$ if the following two conditions are satisfied:
\begin{enumerate}
\item for every $s_{1} \in \mathcal{S}_{1}$ there exists $s_{2} \in \mathcal{S}_{2}$ with $\left(s_{1}, s_{2}\right) \in \mathcal{R}$;
\item for every $\left(s_1, s_2\right) \in \mathcal{R}$ and for every $a_1 \in \mathcal{A}_{1}(s_{1})$ there exists $a_2 \in \mathcal{A}_{2}\left(s_2\right)$ such that for every $s_2^{\prime} \in \Delta_{2}\left(s_2, a_{2}\right)$ there exists $s_1^{\prime} \in \Delta_{1}(s_{1},a_{1})$ satisfying $\left(s_1^{\prime}, s_2^{\prime}\right) \in \mathcal{R}$.
\end{enumerate}
We say that $\Sigma_{1}$ is alternatingly simulated by $\Sigma_{2}$, denoted by $\Sigma_{1} \preceq_{\mathcal{A S}}^{\varepsilon} \Sigma_{2}$, if there exists an alternating simulation relation from $\Sigma_{1}$ to $\Sigma_{2}$.
\end{defn}

As a space discretisation method, this paper proposes to build the Q-learning upon an abstracted model of the system $\Sigma$. An abstraction of system $\Sigma$ involves partitioning the continuous state and input spaces into a finite number of intervals, called cells, using a quantizer, denoted as $\mathrm{q} = (\eta, \mu)$. The parameter $\eta \in \mathbb{R}_{\geqslant 0}$ (respectively, $\mu \in \mathbb{R}_{\geqslant 0}$) represents the interval spacing that is used to approximate the state space (respectively, the input space). 
The resulting abstract system is denoted as the transition system $\Sigma_{\mathrm{D}}$ formally defined by,
\begin{align}\label{Sigma_q}
\Sigma_{\mathrm{D}}=\left(\mathcal{S}_{\mathrm{D}}, \mathcal{A}_{\mathrm{D}}, \Delta, \overline{g}, \underline{g} \right),
\end{align}
where,	
\begin{itemize}
\item The state space $\mathcal{S}_{\mathrm{D}}$ is constructed by discretizing the continuous state space $\mathcal{S}$ into $n_{\xi} \geqslant 1$ states. This discretization is achieved by dividing the state space into intervals or partitions\footnote{In order to define a partition, the sets of measure zero where intervals overlap can be ignored for notational convenience.}, where each element of the partition, denoted by $s$, can be represented as an interval $s = [\xi_1^s, \xi_2^s]$. The parameter $\eta \in \mathbb{R}^+$ is used to control the level of discretization, determining the size of the intervals of the state representation.
\item The input space $\mathcal{A}_{\mathrm{D}}$ is constructed by discretizing the continuous input-space $\mathcal{A}$ into $n_u\geqslant 1$ inputs using a finite partition $\mu \in \mathbb{R}^+$ as a state-space discretization parameter. Each element $a$ of the partition can be described as an interval $a=[v_1^a,v_2^a]$;


\item The transition relation is defined for $s,s^{\prime} \in \mathcal{S}_{\mathrm{D}}$ and $a \in \mathcal{A}_{\mathrm{D}}$ as $s^{\prime}\in \Delta(s,a)$ if and only if,
$$\begin{aligned}
s^{\prime} \cap \left\lbrace f(s_c,a_c) +(L_{f\xi}\eta+L_{fv}\mu )\mathbb{B} \right\rbrace \neq \emptyset,
\end{aligned}$$
where $s_c=\frac{\xi_1^s+\xi_2^s}{2}$ and $a_c=\frac{v_1^a+v_2^a}{2}$;

\item  $\overline{g}(s, a): \; \mathcal{S}_{\mathrm{D}}\times \mathcal{A}_{\mathrm{D}} \rightarrow \mathbb{R} $ is defined by $ \max\limits_{\xi \in s} \max\limits_{v \in a} g(\xi, v) $ and represents the maximal immediate reward achievable when applying action $a$ from the discrete state $s$;

\item $\underline{g}(s, a): \; \mathcal{S}_{\mathrm{D}}\times \mathcal{A}_{\mathrm{D}} \rightarrow \mathbb{R} $ is defined by $ \min\limits_{\xi \in s}  \min\limits_{v \in a}  g(\xi, v) $ and represents the minimal immediate reward achievable when applying action $a$ from the discrete state $s$;
\end{itemize}
For the transition system $\Sigma_{\mathrm{D}}$, the set of enabled inputs for a state $s \in \mathcal{S}_{\mathrm{D}}$ is formally defined by,
\begin{align}\label{key9}
\resizebox{0.85\hsize}{!}{$a \in \mathcal{A}_{\mathrm{D}}(s) \Longleftrightarrow 	s^{\prime} \cap \left\lbrace f(s_c,a_c) +(L_{f\xi}\eta+L_{fv}\mu )\mathbb{B} \right\rbrace \subseteq \mathcal{S}_{\mathrm{D}},$}
\end{align}
where $s_c=\frac{\xi_1^s+\xi_2^s}{2}$ and $a_c=\frac{v_1^a+v_2^a}{2}$, and
\begin{align}\label{key10}
\resizebox{0.85\hsize}{!}{$\forall a \in \mathcal{A}_{\mathrm{D}}(s), \Delta(s,a)=s^{\prime} \cap \left\lbrace f(s_c,a_c) +(L_{f\xi}\eta+L_{fv}\mu )\mathbb{B}\right\rbrace.$}
\end{align}
The following proposition relates formally system $\Sigma_{\mathrm{D}}$ to system $\Sigma$ by an alternating simulation relation.

\begin{prop}\label{Prop1}
Consider the transition systems $\Sigma = (\mathcal{S}, \mathcal{A}, F, g, g)$ and $\Sigma_{\mathrm{D}} = \left(\mathcal{S}_{\mathrm{D}}, \mathcal{A}_{\mathrm{D}}, \Delta, \overline{g}, \underline{g}\right)$ then, the relation,  $$ \mathcal{R} =\left\lbrace (\xi, s) \in \mathcal{S}\times \mathcal{S}_{\mathrm{D}}| \;\; \xi\in s \right\rbrace,$$ is an  alternating simulation relation from $\Sigma_{\mathrm{D}}$ to system $\Sigma$.
\end{prop}

\begin{proof}
The first condition of alternating simulation follow directly from the form of $\mathcal{R}$ and from the fact that $\Sigma_{\mathrm{D}}$ and $\Sigma$ have the same sets $\mathcal{S}= \cup_{\ell=1}^{M} s_{i}, \; s_{i}\in \mathcal{S}_{\mathrm{D}}$.

Since $\mathcal{S} \subseteq \mathcal{S}_{\mathrm{D}}$, it follows from the definition of the transition relation $F$ that,
$$
\forall \xi \in s, v \in  \mathcal{A}(\xi), f(\xi, v) \subseteq \Delta(s,a).
$$
Therefore, from \eqref{key7} and \eqref{key9}, it follows that we have for all $s \in \mathcal{S}_{\mathrm{D}}$, $\mathcal{A}_{\mathrm{D}}(s) \subseteq \mathcal{A}(s) $. Moreover, from \eqref{key8} and \eqref{key10}, for all $a\in \mathcal{A}_{\mathrm{D}}(s)$, it holds $F(s,a)\subseteq\Delta(s,a)$.

Let us now show the second condition of alternating simulation. Consider $\left(\xi, s \right) \in \mathcal{R}$, take any $\xi \in s$. Let $a \in \mathcal{A}_{\mathrm{D}}\left(s\right)$, then there exists $v\in a$, such that $v \in \mathcal{A}\left(\xi\right)$. Note that such $\xi$ and $v$ satisfy, $\|\xi-\xi_{c}\|\leqslant\eta$ and $\|v-v_{c}\|\leqslant\mu$. Hence, because $f\left(\xi, v\right)-f\left(s_{c}, v_{c}\right) \in (L_{f\xi}\eta+L_{fv}\mu )\mathbb{B} $, we have $F\left(\xi, v\right) \subseteq \Delta\left(s,a\right)$.
Therefore, for all $\xi^{\prime} = F\left(\xi,v\right)$, there exists $s^{\prime} \in\Delta\left(s, a\right)$ satisfying $\xi' \in s'$ and hence $(\xi',s')\in \mathcal{R}$.
\end{proof}


It is important to highlight that within the abstracted model $\Sigma_{\mathrm{D}}$, two reward functions are introduced to capture the maximal and minimal reward over a discrete state $s$. The work of \cite{watkins1989learning} has shown that one of the conditions for convergence of Q-learning is a bounded reward function. In this work, for the sake of generality, the following assumption is made,
\begin{assum}\label{Assum2}
For any $\xi, \overline{\xi} \in \mathcal{S}$, $v \in \mathcal{A}(\xi)$, and $\overline{v} \in \mathcal{A}(\overline{\xi})$, the reward function $g$ satisfies:
\begin{align}\label{EqAssum2}
| g(\xi,v) - g(\overline{\xi},\overline{v}) | \leqslant  L_{g \xi}\|\xi-\overline{\xi}\|+L_{g v}\|v-\overline{v}\|,
\end{align}
where $L_{g\xi}$ and $L_{gv}$ are positive constants.
\end{assum}

\section{Q-learning for symbolic model}
Inspired by the construction in \cite{borri2023reinforcement}, this paper proposes to build the Q-learning on the reward structure in \eqref{Sigma_q} and defines the maximal and minimal $return$ by, 
\begin{align*}
&\overline{G}_{k}\!\!=\!\!\overline{g}(\!s_{k+1},a_{k+1})\!\!+\!\!\gamma \overline{g}(s_{k+2},a_{k+2})\!\!+\!\!\gamma^{2} \overline{g}(\!s_{k+3},a_{k+3})\!\!+\!\ldots \\	&\underline{G}_{k}\!\!=\!\!\underline{g}(\!s_{k+1},a_{k+1})\!\!+\!\!\gamma \underline{g}(\!s_{k+2},a_{k+2})\!\!+\!\!\gamma^{2} \underline{g}(\!s_{k+3},a_{k+3})\!\!+\!\ldots
\end{align*}
because $ \overline{G}_{k}=\sum_{i=0}^{\infty}\gamma^{i} \overline{g}(s_{k+i+1},a_{k+i+1}) $ resp. ($ \underline{G}_{k}=\sum_{i=0}^{\infty}\gamma^{i} \underline{g}(s_{k+i+1},a_{k+i+1}) $) and $\overline{G}_{k+1}=\sum_{i=0}^{\infty} \gamma^{i} \overline{G}_{k+i+2}$ resp. ($\underline{G}_{k+1}=\sum_{i=0}^{\infty} \gamma^{i} \overline{G}_{k+i+2}$), both satisfy the \textit{consistency} relation,
\begin{equation}\label{Eq3}
\begin{array}{c}
\overline{G}_{k}=\overline{g}(s_{k+1},a_{k+1})+\gamma \overline{G}_{k+1},\\
\underline{G}_{k}=\underline{g}(s_{k+1},a_{k+1})+\gamma \underline{G}_{k+1}.
\end{array}
\end{equation}
where $\gamma \in(0,1]$ is the discount rate  and the tail of a state run starting at time $k \geq 0$, together with the corresponding sequence of rewards is given by, 
$$
s_{k} \underset{g(s_{k+1},a_{k+1})}{\stackrel{\mathcal{A}(s_{k})}{\longrightarrow}} s_{k+1} \underset{g(s_{k+2},a_{k+2})}{\stackrel{\mathcal{A}(s_{k+2})}{\longrightarrow}} s_{k+2} \underset{g(s_{k+3},a_{k+3})}{\stackrel{\mathcal{A}(s_{k+3})}{\longrightarrow}} s_{k+3} \cdots
$$
While there may be some similarities between certain aspects of this section and the work by \cite{borri2023reinforcement}, it is important to highlight that the approach presented here is entirely novel. This paper introduces a distinct construction that leverages the concept of maximally and minimally return, which generalizes the one in \cite{borri2023reinforcement}.

\subsection{Maximal and Minimal State Value Functions under a Policy $\pi$}
The maximal and minimal state value functions of a state $s \in S_\mathrm{D}$ under a policy $\pi$, denoted $\overline{v}_{\pi}(s)$ and $\underline{v}_{\pi}(s)$, respectively, are defined as the maximal and minimal return that can be obtained by starting in state $s$ and following policy $\pi$.
$$
\begin{aligned}
	& \bar{v}_{\pi}(s)=\max \left\{\overline{G}_{k} \mid s=s_{k}, \pi(s_{\tau})\in \mathcal{A}_{\mathrm{D}}, \tau \geqslant k\right\}, \\
	& \underline{v}_{\pi}(s)=\min \left\{\underline{G}_{k} \mid s=s_{k}, \pi(s_{\tau})\in \mathcal{A}_{\mathrm{D}}, \tau \geqslant k\right\} .
\end{aligned}
$$
The Bellman equations for the maximal and minimal state value functions can be derived directly from equation \eqref{Eq3} and the definitions of $\overline{v}_{\pi}$ and $\underline{v}_{\pi}$, as shown in the following proposition.
\begin{prop}
For each $(s, a) \in S_{\mathrm{D}} \times \mathcal{A}_{\mathrm{D}}$, define $\Delta_{\pi}(s)=\Delta(s, \pi(s))$, for all $s \in S_{\mathrm{D}}$. If the maximal and minimal value functions under policy $\pi$ exist, then they satisfy the following equations,
\begin{subequations}\label{Eq4}
\begin{gather}
\overline{v}_{\pi}(s)=\overline{g}(s, \pi(s))+\gamma \max _{s^{\prime} \in \Delta_{\pi}(s)} \overline{v}_{\pi}\left(s^{\prime}\right), \\
\underline{v}_{\pi}(s)=\underline{g}(s, \pi(s))+\gamma \min _{s^{\prime} \in \Delta_{\pi}(s)} \underline{v}_{\pi}\left(s^{\prime}\right).
\end{gather}
\end{subequations}
\end{prop}

\begin{cor}
If $\gamma \in(0,1)$, then the state value functions $\overline{v}_{\pi}$ and $\underline{v}_{\pi}$ exist, are unique, and satisfy the Bellman equations \eqref{Eq3}.
\end{cor}
\begin{pf}
This proof has been omitted because it follows similar steps used to prove \cite[Corollary 1]{borri2023reinforcement} with the difference that each state value function is defined with its corresponding discrete reward function $\overline{g}$ or $\underline{g}$. \hfill $\square$
\end{pf}

It is worth noting that in the context of non-deterministic transition system, two state value functions are need to be defined instead of a single one for two reasons. Firstly, for each state, there are two reward functions, $\overline{g}$ and $\underline{g}$, to capture all possible rewards when transiting to the successor state $ s' \in  \Delta_{\pi}(s) \subset S_{\mathrm{D}}$. Secondly, because the successor of a discrete state is not necessarily unique, these reward functions permits to identify the range of rewards of the successors.

\begin{rem}
The relation between the Bellman equations \eqref{Eq4} and the classical optimality Bellman equation for Markov Decision Processes (MDPs), indicates that dynamic programming can be employed to determine the maximal and minimal state value functions under a policy $\pi$. It is important to highlight that the iterative application of the Bellman equations \eqref{Eq4} exhibits $\gamma$-contraction, leading to linear convergence of these algorithms at a rate of $\gamma$. \hfill $\diamond$
\end{rem}
\subsection{Maximal and Minimal $q$-Functions under a Policy $\! \pi$}

Now, define the maximal and  minimal state-action value function for state $s \in S_\mathrm{D}$ and action $a \in \mathcal{A}_{\mathrm{D}} $ under policy $\pi$ by,
\begin{subequations}
\begin{gather}
\overline{q}_{\pi}(s, a)=\max \left\{\overline{G}_{k} \mid s_{k}=s, a_{k}=a,  \pi(s^{\prime})\right\}, \\
\underline{q}_{\pi}(s, a)=\min \left\{\underline{G}_{k} \mid s_{k}=s, a_{k}=a,  \pi(s^{\prime})\right\}.
\end{gather}
\end{subequations}
where $\overline{q}_{\pi}(s, a)$ and $\underline{q}_{\pi}(s, a)$, denote the maximal and minimal return that can be obtained starting in state $s$, taking action $a$, and following $\pi$ thereafter.

The following proposition is straightforward from the consistency relation \eqref{Eq3} and the definitions of $\overline{v}_{\pi}, \overline{q}_{\pi} \cdot \underline{v}_{\pi}$, and $\underline{q}_{\pi}$.

\begin{prop}
If $\gamma \in(0,1)$, then $\overline{q}_{\pi}$ and $\underline{q}_{\pi}$ exist, are unique, and satisfy the following consistency relations
\begin{subequations}
\begin{gather}
\overline{q}_{\pi}(s, a)=\overline{g}(s, \pi(s))+\gamma \max _{s^{\prime} \in \Delta(s, a)} \overline{v}_{\pi}\left(s^{\prime}\right) ,\\
\underline{q}_{\pi}(s, a)=\underline{g}(s, \pi(s))+\gamma \min _{s^{\prime} \in \Delta(s, a)} \underline{v}_{\pi}\left(s^{\prime}\right).
\end{gather}
\end{subequations}
\end{prop}

In view of (3) and (4), one has that $\overline{v}_{\pi}(s)=\overline{q}_{\pi}(s, \pi(s))$ and $\underline{v}_{\pi}(s)=\underline{q}_{\pi}(s, \pi(s))$. Therefore, combining these equations, one can obtain Bellman equations for $\overline{q}_{\pi}$ and $\underline{q}_{\pi}$,
\begin{subequations}\label{Eq6}
\begin{gather}
\resizebox{0.87\hsize}{!}{$\overline{q}_{\pi}(s, a)=\overline{g}(s, \pi(s))+\gamma \max\limits_{s^{\prime} \in \Delta(s, a)} \max\limits_{a^{\prime} \in \mathcal{A}(s^{\prime})} \overline{q}_{\pi}\left(s^{\prime}, a^{\prime}\right),$} \\
\resizebox{0.87\hsize}{!}{$\underline{q}_{\pi}(s, a)=\underline{g}(s, \pi(s))+\gamma \!\!\min\limits_{s^{\prime} \in \Delta(s, a)} \max\limits_{a^{\prime} \in \mathcal{A}(s^{\prime})} \underline{q}_{\pi}\left(s^{\prime}, a^{\prime}\right)$} .
\end{gather}
\end{subequations}


\subsection{The optimal state action value function}
The maximally and minimally optimal value functions are defined as follows:
\begin{eqnarray}
\forall s \in S_{\mathrm{D}}, \: \overline{v}^{\star}(s) = \max\limits_{\pi} \overline{v}_{\pi}(s), \: \underline{v}^{\star}(s) = \max\limits_{\pi} \underline{v}_{\pi}(s) 
\end{eqnarray}
where all maximally and minimally optimal policies $\overline{\pi}^{\star}\left(\underline{\pi}^{\star}\right)$ have their own maximal optimal state-value function $\overline{v}^{\star}\left(\underline{v}^{\star}\right)$. Using these functions, the maximally and minimally optimal state-action value functions are defined by,
\begin{subequations}
\begin{gather}
\overline{q}^{\star}(s, a)=\max\limits_{\pi} \overline{q}_{\pi}(s, a),\\
\underline{q}^{\star}(s, a)=\max\limits_{\pi} \underline{q}_{\pi}(s, a),
\end{gather}
\end{subequations}
where $\overline{q}^{\star}(s, a)$ is the maximal and minimal return that can be obtained starting in state $s$, taking action $a$, and acting optimally thereafter.
These functions satisfy the following equations:
\begin{subequations}\label{Eq18ab}
\begin{gather}
\overline{q}^{\star}(s, a)=\overline{g}(s, a)+\gamma \max _{s^{\prime} \in \Delta(s, a)}\max _{a^{\prime} \in \mathcal{A}(s^{\prime})} \overline{q}^{\star}\left(s^{\prime}\right), \\
\underline{q}^{\star}(s, a)=\underline{g}(s, a)+\gamma \min _{s^{\prime} \in \Delta(s, a)}\max _{a^{\prime} \in \mathcal{A}(s^{\prime})} \underline{q}^{\star}\left(s^{\prime}\right) .
\end{gather}
\end{subequations}
The following theorem characterizes the existence and uniqueness of the maximally and minimally optimal state value functions of a symbolic model $\Sigma_{\mathrm{D}}$.
\begin{thm}\label{THM2}
If $\gamma \in(0,1)$, then there exists a unique maximally and minimally optimal state value function $\overline{v}^{\star}$ and a unique minimally optimal state value function $\underline{v}^{\star}$ that satisfy,
\begin{subequations}\label{Eq19ab}
\begin{gather}
\overline{v}^{\star}(s)=\max _{a \in \mathcal{A}_{\mathrm{D}}}\left\{\overline{g}(s, a)+\gamma \max _{s^{\prime} \in \Delta(s, a)} \overline{v}^{\star}\left(s^{\prime}\right)\right\}, \\
\underline{v}^{\star}(s)=\max _{a \in \mathcal{A}_{\mathrm{D}}}\left\{\underline{g}(s, a)+\gamma \min _{s^{\prime} \in \Delta(s, a)} \underline{v}^{\star}\left(s^{\prime}\right)\right\} .
\end{gather}
\end{subequations}
\end{thm}
\begin{pf}
\textcolor{black}{
By the Bellman optimality principle, optimal value functions involve taking the optimal action at the first step and acting optimally thereafter. This leads to consistency relations:
\begin{align*}
&\bar{v}^{\star}(s) = \max_{a \in \mathcal{A}}\left\{\overline{g}(s, a) + \gamma \max_{s' \in \Delta(s, a)} \bar{v}^{\star}\left(s'\right)\right\} ,\\
&\underline{v}^{\star}(s) = \max_{a \in \mathcal{A}}\left\{\underline{g}(s, a) + \gamma \min_{s' \in \Delta(s, a)} \underline{v}^{\star}\left(s'\right)\right\}.
\end{align*}
Define $\overline{T}^{\star} V(s) = \max_{a \in \mathcal{A}}\left\{r(s, a) + \gamma \max_{s' \in \Delta(s, a)} V\left(s'\right)\right\}$ and $\underline{T}^{\star} V(s) = \max_{a \in \mathcal{A}}\left\{r(s, a) + \gamma \min_{s' \in \Delta(s, a)} V\left(s'\right)\right\}$, as operators. Then, for value functions \(V\) and \(U\):
\begin{align*}
&\left\|\overline{T}^{\star} V - \overline{T}^{\star} U\right\|_{\infty} \leq \gamma \|V - U\|_{\infty}, \\
&\left\|\underline{T}^{\star} V - \underline{T}^{\star} U\right\|_{\infty} \leq \gamma \|V - U\|_{\infty}.
\end{align*}		
Thus, both operators \(\overline{T}^{\star}\) and \(\underline{T}^{\star}\) are \(\gamma\)-contractions. Therefore, by the Banach fixed-point theorem, if \(\gamma \in (0,1)\), there exists a unique maximally (minimally) optimal state value function \(\bar{v}^{\star}(\underline{v}^{\star})\) satisfying \eqref{Eq19ab}.}
\hfill $\square$
\end{pf}
\begin{cor}\label{Cor}
If $\gamma \in(0,1)$, then there exist deterministic maximally and minimally optimal policies for the symbolic model $\Sigma_{\mathrm{D}}$. In particular, these policies are given by,
\begin{eqnarray}\label{Eq20ab}
\overline{\pi}^{\star}(s)=\underset{a \in \mathcal{A}}{\arg \max } \; \overline{q}^{\star}(s, a), \:  \underline{\pi}^{\star}(s)=\underset{a \in \mathcal{A}}{\arg \max } \; \underline{q}^{\star}(s, a) .
\end{eqnarray}
\end{cor}
\begin{pf}\textcolor{black}{
By Theorem \ref{THM2}, if $\gamma \in(0,1)$, then the state-action value function $\bar{q}^{\star}\left(\underline{q}^{\star}\right)$ exists, is unique, and can be obtained by $\overline{v}^{\star}$ $\left(\underline{v}^{\star}\right)$ using \eqref{Eq18ab}. Hence, by the Bellman optimality principle, the deterministic policy $\bar{\pi}^{\star}\left(\pi^{\star}\right)$ given in \eqref{Eq20ab} is maximally (resp. minimally) optimal.}
\end{pf}

\begin{rem}
The work presented in \cite{borri2023reinforcement} was motivated by the question of applying Q-learning to nondeterministic transition systems that lack transition probabilities. However, it is important to note that their focus was not specifically on addressing the Q-learning problem for continuous dynamical systems. They primarily revolves around the convergence of the maximally and minimally Q-values for non-deterministic transition systems. However, their work does not provide any guarantees for applying abstraction-based controller synthesis to systems with continuous state-action space. In contrast, the current work addresses a different problem involving discrete-time systems with continuous state and action spaces. The object is to obtain an equivalent abstracted model that preserves the reward structure within the discrete space. In this approach, two reward functions are used to learn the maximally and minimally Q-values, which changes the entire construction done in \cite{borri2023reinforcement}. 
\hfill $\diamond$
\end{rem}

%

\section{Balancing conservatism with optimal policy approximation in abstraction}
The goal now is to relate the Q-values of the original system $\Sigma$ with the Q-values of the symbolic model $\Sigma_{\mathrm{D}}$. The following assumption regarding the initialisation of the Q-values of $\underline{q}_{\pi}^{(0)}$, $q_{\pi}^{(0)}$ $\overline{q}_{\pi}^{(0)}$ is required to derive the main results.
\begin{assum}\label{Assum4}
The initial Q-values, denoted as $\underline{q}_{\pi}^{(0)}$, $q_{\pi}^{(0)}$, and $\overline{q}_{\pi}^{(0)}$ are equal across all states for each action, i.e., $\underline{q}^0(s,a)=\underline{q}^0(s',a)$, $q^{0}(\xi,v)=q^{0}(\xi',v)$, $\underline{q}^{0}(s,a)=\underline{q}^{0}(s',a)$, $s,s' \in \mathcal{S}_{\mathrm{D}}, \xi,\xi' \in \mathcal{S}$,  and satisfy the following conditions:
\begin{align}\label{EqAssum4}
\underline{q}_{\pi}^{(0)} \leqslant q_{\pi}^{(0)} \leqslant \overline{q}_{\pi}^{(0)}
\end{align}
\end{assum}
Note that Assumption \ref{Assum4} is not necessary for convergence of $\overline{q}_{\pi}$ and $\underline{q}_{\pi}$ but serves the purpose of approximating the continuous action value function. The following result show that, under Assumption \ref{Assum4}, the true Q-values remain bounded between $\underline{q}_{\pi}^{(k)}$ and $\overline{q}_{\pi}^{(k)}$ for $\forall k \in \mathbb{N}$.
\begin{thm}\label{THM3}
Under Assumption \ref{Assum4} is satisfied. Let $\xi \in \mathcal{S} $, $s \in S_{\mathrm{D}} $, $v\in \mathcal{A}$, $a\in \mathcal{A}_{\mathrm{D}}$ and a policy $\pi \in \Pi$, let $q(\xi,v)$,  denotes the true Q-function of a continuous state space model $\Sigma$ and let $\underline{q}_{\pi}(s,a)$ and $\overline{q}_{\pi}(s,a)$ be the maximal and minimal Q-values, respectively, obtained by applying the update law \eqref{Eq6} to the symbolic model $\Sigma_{\mathrm{D}}$. Then, the true Q-values satisfies:
\begin{eqnarray}\label{Eq14}
\underline{q}_{\pi}^{(k)}(s,a) \leqslant  q_{\pi}^{(k)}(\xi,v) \leqslant  \overline{q}_{\pi}^{(k)}(s,a).
\end{eqnarray}
\end{thm}
\begin{pf}
The proof follows an induction-based approach. When $k = 0$, we have $q_{\pi}^{(0)}(s,a) \leqslant \underline{q}_{\pi}^{(0)}(s,a) \leqslant \overline{q}_{\pi}^{(0)}(s,a)$, by means of Assumption \ref{Assum4} the statement is true. Assume that $\underline{q}_{\pi}^{(k-1)}(s,a) \leqslant  q_{\pi}^{(k-1)}(\xi,v) \leqslant  \overline{q}_{\pi}^{(k-1)}(s,a)$. We want to show that $\underline{q}_{\pi}^{(k)}(s,a) \leqslant  q_{\pi}^{(k)}(\xi,v) \leqslant  \overline{q}_{\pi}^{(k)}(s,a)$ also holds for all $\xi \in s, \; v\in a$ pairs. From the update equations \eqref{update} and \eqref{Eq6}, we have: 
\begin{align*}
\underline{q}_{\pi}^{(k)}(s,a)=& \underline{g}(s,a) +\gamma \min\limits_{s^{\prime}\in \Delta(s,a)} \max\limits_{a^{\prime}\in\mathcal{A}_{\mathrm{D}}(s')}\underline{q}_{\pi}^{(k-1)}(s^{\prime},a^{\prime}),\\
q_{\pi}^{(k)}(\xi,v)=& g(\xi,v) + \gamma  \max\limits_{b^{\prime}\in\mathcal{A}(\xi^{\prime})}q_{\pi}^{(k-1)}(\xi^{\prime},b^{\prime}),\\
\overline{q}_{\pi}^{(k)}(s,a)=&\overline{g}(s,a)+\gamma \max\limits_{s^{\prime}\in \Delta(s,a)} \max\limits_{a^{\prime}\in\mathcal{A}_{\mathrm{D}}(s')}\overline{q}_{\pi}^{(k-1)}(s^{\prime},a^{\prime}).
\end{align*}
The definitions of $ \underline{g}(s,a) $ and $ \overline{g}(s,a) $ in \eqref{Sigma_q} gives $ \underline{g}(s,a) \leqslant  g(\xi,v) \leqslant  \overline{g}(s,a)$. Based on Proposition \ref{Prop1}, we can observe that since $v' \in \mathcal{A}(\xi')$, there exists $a' \in \mathcal{A}_{\mathrm{D}}(s')$ such that $v' \in a'$. Additionally, as $\xi' = f(\xi, v)$, there exists $s' \in \Delta(s, a)$ such that $\xi' \in s'$. By leveraging the assumption at time $k-1$, we can establish the inequality:
$$ \begin{cases}
\max_{a' \in \mathcal{A}_{\mathrm{D}}(s')} \underline{q}_{\pi}^{(k-1)}(s',a') \leqslant \max\limits_{v' \in \mathcal{A}(\xi')} q_{\pi}^{k-1}(\xi',v') \\
\max_{a' \in \mathcal{A}_{\mathrm{D}}(s')} \overline{q}_{\pi}^{(k-1)}(s',a') \geqslant \max\limits_{v' \in \mathcal{A}(\xi')} q_{\pi}^{k-1}(\xi',v')
\end{cases}, $$ which implies that, $$
\resizebox{0.95\hsize}{!}{$\begin{cases}  \min\limits_{s^{\prime}\in \Delta(s,a)} \max_{a' \in \mathcal{A}_{\mathrm{D}}(s')} \underline{q}_{\pi}^{(k-1)}(s',a') \leqslant \max\limits_{v' \in \mathcal{A}(\xi')} q_{\pi}^{k-1}(\xi',v')\\
\max\limits_{s^{\prime}\in \Delta(s,a)} \max_{a' \in \mathcal{A}_{\mathrm{D}}(s')} \overline{q}_{\pi}^{(k-1)}(s',a') \geqslant \max\limits_{v' \in \mathcal{A}(\xi')} q_{\pi}^{k-1}(\xi',v')
\end{cases}.$}	 $$
Thus, for all $k \geqslant 0$, $\underline{q}_{\pi}^{(k)}(s,a) \leqslant  q_{\pi}^{(k)}(\xi,v) \leqslant  \overline{q}_{\pi}^{(k)}(s,a)$ holds, for all $s \in S_{\mathrm{D}}$ and $a \in \mathcal{A}_{q}$. \hfill$\square$
\end{pf}
The result established in Theorem \ref{THM3} is illustrated in Figure \ref{fig:optimalq}. It is shown that the range of the optimal Q-values of system \eqref{Eq1} with continuous state space is bounded between $\underline{q}_{\pi}^{\ast}(s,a)$ and $\overline{q}_{\pi}^{\ast}(s,a)$.
\begin{figure}
\centering
\includegraphics[width=0.7\linewidth]{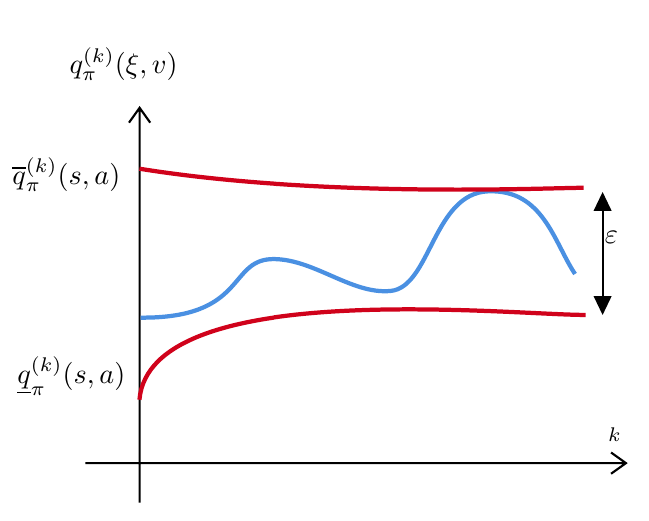}
\caption{The $k^{th}$ Q-value function for a given symbolic state $s$ over the continuous state space containing $\xi$.}
\label{fig:optimalq}
\end{figure}

The following result shows that the continuous Q-values in \eqref{update} satisfies a certain Lipschitz property. 
\begin{thm}
If Assumptions \ref{Assum1}, \ref{Assum2} and \ref{Assum4} are satisfied, then $\forall \xi, \overline{\xi} \in \mathcal{S}$ and $v \in \mathcal{A}(\xi), \overline{v} \in \mathcal{A}(\overline{\xi})$ the Q-value function determined by the update rule \eqref{update} satisfies the following Lipschitz continuity property for all $k\in \mathbb{N}$,
\begin{align}\label{Eq15}
\left|q_{\pi}^{(k)}(\xi, v)-q_{\pi}^{(k)}\left(\overline{\xi}, \overline{v}\right)\right| \leqslant  L_{\xi}^{(k)}\! \|\xi-\overline{\xi}\|\!+ \! L_{v}^{(k)}\|v-\overline{v}\|,
\end{align}
where $ L_{\xi}^{(k)}$ and $L_{v}^{(k)}$ are defined by,
\begin{subequations}\label{Eq20}
\begin{gather}\label{Eq20a}
L_{\xi}^{(k)}=\gamma L_{\xi}^{(k-1)}L_{f \xi}+\gamma L_{v}^{(k-1)} L_{\mathcal{A}}L_{f \xi}+L_{g\xi},\\\label{Eq20b}
L_{v}^{(k)}=\gamma L_{\xi}^{(k-1)}L_{f v}+\gamma L_{v}^{(k-1)} L_{\mathcal{A}}L_{f v}+L_{gv}.
\end{gather}
\end{subequations}

\end{thm}

\begin{pf}
The proof follows an induction-based approach. When $k=1$, the update rule \eqref{update} and condition \eqref{EqAssum2} show that for any $\xi, \overline{\xi} \in \mathcal{S}$ and $ v \in \mathcal{A}(\xi)$, $\overline{v} \in \mathcal{A}(\overline{\xi})$ we have, 	
\begin{align}\label{Eq21;init}
\left|q_{\pi}^{(1)}(\xi, v)-q_{\pi}^{(1)}(\overline{\xi}, \overline{v})\right| \leqslant  |g(\xi,v)-g(\overline{\xi}, \overline{v})| \nonumber\\
+  \gamma \left| \max\limits_{v \in \mathcal{A}(\xi)} q_{\pi}^{(0)}(\xi,v) - \max\limits_{\overline{v} \in \mathcal{A}(\overline{\xi})} q_{\pi}^{(0)}(\overline{\xi},\overline{v}) \right|.
\end{align}
The first term on the right and side of \eqref{Eq21;init} implies that, 
\begin{align}
|g(\xi,v)-g(\overline{\xi}, \overline{v})|\leqslant L_{g \xi}\|\xi-\overline{\xi}\|+L_{g v}\|v-\overline{v}\|.
\end{align}
Under Assumption \ref{Assum4}, the second term yields, 
\begin{align}
\left| \max\limits_{v \in \mathcal{A}(\xi)} q_{\pi}^{(0)}(\xi,v) - \max\limits_{\overline{v} \in \mathcal{A}(\overline{\xi})} q_{\pi}^{(0)}(\overline{\xi},\overline{v}) \right| =0 .
\end{align}
Thus, it follows that, 
\begin{align}\label{Eq21}
\left|q_{\pi}^{(1)}(\xi, v)-q_{\pi}^{(1)}(\overline{\xi}, \overline{v})\right| \leqslant   L_{g \xi}\|\xi-\overline{\xi}\|+L_{g v}\|v-\overline{v}\|.
\end{align}
Suppose that for every $\xi, \overline{\xi} \in \mathcal{S}$, $v \in \mathcal{A}(\xi)$, $\overline{v} \in \mathcal{A}(\overline{\xi})$, 
\begin{align}\label{Eq17}
\left|q_{\pi}^{(k-1)}\!\!(\xi,\! v)\!-\! q_{\pi}^{(k-1)}\!\!\left(\overline{\xi}, \!\overline{v}\right)\!\right| \!\! \leqslant  \!\! L_{\xi}^{(k-1)}\!\!\left\|\xi-\overline{\xi}\right\| \!\!+\!\! L_{v}^{(k-1)}\!\!\left\|v-\overline{v}\right\|,
\end{align}
holds for all $(\xi,v)$ pairs. To show that \eqref{Eq17} holds for $k$, consider the update law \eqref{update}. Using the triangular inequality, we have, 
\begin{align}\label{Eq18}
&\hspace{-10pt}\left|q_{\pi}^{(k)}(\xi, v)-q_{\pi}^{(k)}\left(\overline{\xi}, \overline{v}\right)\right| \leqslant  \left|g(\xi,v) -g(\overline{\xi},\overline{v})\right| \nonumber\\
&\hspace{10pt}+\gamma\left|\max\limits_{v^{\prime} \in \mathcal{A}(\xi^{\prime})}\!\! q_{\pi}^{(k-1)}(\xi^{\prime},v^{\prime})-\!\!\!\!\max\limits_{\overline{v} \in \mathcal{A}(\overline{\xi}^{\prime})} q_{\pi}^{(k-1)}(\overline{\xi}^{\prime},\overline{v}^{\prime}) \right|.
\end{align}
Because $\mathcal{A}(\xi^{\prime})$ and $\mathcal{A}(\overline{\xi}^{\prime})$ are compact sets, there exists $v^{\prime} \in \mathcal{A}(\xi^{\prime})$ s.t. $q_{\pi}^{(k-1)}(\xi^{\prime}, v^{\prime})=\max\limits_{v^{\prime} \in \mathcal{A}\left(\xi^{\prime}\right)} q_{\pi}^{(k-1)}(\xi^{\prime}, v^{\prime})$ and $\overline{v}^{\prime} \in \mathcal{A}(\overline{\xi}^{\prime})$ s.t. $q_{\pi}^{(k-1)}(\overline{\xi}^{\prime}, \overline{v}^{\prime})=\max\limits_{\overline{v}^{\prime} \in \mathcal{A}\left(\overline{\xi}^{\prime}\right)} q_{\pi}^{(k-1)}(\overline{\xi}^{\prime}, \overline{v}^{\prime})$, the right hand side of \eqref{Eq18} implies the chain of inequalities \eqref{Eq19}.
\begin{align}\label{Eq19}
	&\left|q_{\pi}^{(k)}(\xi, v)-q_{\pi}^{(k)}\left(\overline{\xi}, \overline{a}\right)\right| \leqslant  L_{g\xi} \|\xi-\overline{\xi}\|+L_{gv}\|v-\overline{v}\| \nonumber\\ &+\gamma\!\left(q_{\pi}^{(k-1)}(\xi^{\prime}, v^{\prime})\!-\! q_{\pi}^{(k-1)}(\overline{\xi}^{\prime}, \overline{v}^{\prime})\!\right)\nonumber\\
	\leqslant & L_{g\xi} \|\xi-\overline{\xi}\|+L_{gv}\|v-\overline{v}\| \nonumber\\ &+\gamma \left(L_{\xi}^{(k-1)}\|\xi^{\prime}-\overline{\xi}^{\prime}\| +L_{v}^{(k-1)}\|v^{\prime}-\overline{v}^{\prime}\|\right)\nonumber\\
	\leqslant &  L_{g\xi} \|\xi-\overline{\xi}\|+L_{gv}\|v-\overline{v}\| \nonumber\\ 
 &+ \gamma\left( L_{\xi}^{(k-1)}(L_{f \xi}\|\xi-\overline{\xi}\|+L_{f v}\|v-\overline{v}\|) \right.\nonumber\\
 &+ \left. \!  \gamma L_{v}^{(k-1)}  L_{\mathcal A}\|\xi^{\prime}\!-\!\overline{\xi}^{\prime}\!\|\right) \nonumber\\
	\leqslant &  L_{g\xi} \|\xi-\overline{\xi}\|+L_{gv}\|v-\overline{v}\| \nonumber\\ 
	&+ (\gamma L_{\xi}^{(k-1)}L_{f \xi}+\gamma L_{v}^{(k-1)}  L_{\mathcal A}L_{f \xi}) \|\xi-\overline{\xi}\|  \nonumber\\
 &+ (\gamma L_{\xi}^{(k-1)}L_{f v}+\gamma L_{v}^{(k-1)}  L_{\mathcal A}L_{f v}) \|v-\overline{v}\|  \nonumber\\
	= &  (\gamma L_{\xi}^{(k-1)}L_{f \xi}+\gamma L_{v}^{(k-1)}  L_{\mathcal A}L_{f \xi}+L_{g\xi}) \|\xi-\overline{\xi}\|\nonumber\\ &+ (\gamma L_{\xi}^{(k-1)}L_{f v}+\gamma L_{v}^{(k-1)}  L_{\mathcal A}L_{f v}+L_{gv}) \|v-\overline{v}\| \nonumber\\
	= & L_{\xi}^{(k)}\! \|\xi-\overline{\xi}\|\!+ \! L_{v}^{(k)}\|v-\overline{v}\|.
\end{align}
The third inequality is obtained by utilizing \eqref{EqAssum1} and \eqref{EqLemma1}, while the fourth inequality is derived from the repeated use of \eqref{EqAssum1}. Thus, for any $\xi, \overline{\xi} \in \mathcal{S}$ and action pairs $v \in \mathcal{A}(\xi)$, $\overline{v} \in \mathcal{A}(\overline{\xi})$, we have, 
$\left|q_{\pi}^{(k)}(\xi, v)-q_{\pi}^{(k)}\left(\overline{\xi}, \overline{v}\right)\right| \leqslant  L_{\xi}^{(k)}\! \|\xi-\overline{\xi}\|\!+ \! L_{v}^{(k)}\|v-\overline{v}\|$. By induction, this holds for all $k\geqslant 0 $. \hfill $\square$
\end{pf}

The following result make it possible to measure the mismatch between the Q-values of the symbolic model in \eqref{Eq6} as a function of the discretisation parameters $\eta$ and $\mu$.


\begin{thm}\label{THM5}
Consider the symbolic model \eqref{Sigma_q} with positive constants $\eta$ and $\mu$. If Assumptions \ref{Assum1}, \ref{Assum2} and \ref{Assum4} are satisfied, then for every $k \in \mathbb{N}$, $\underline{q}(s,a)$ and $ \overline{q}(s,a)$ satisfy,
\begin{align}\label{Eq23}
\left|\underline{q}_{\pi}^{(k)}(s,a) - \overline{q}_{\pi}^{(k)}(s,a)\right|\leqslant   L_{\xi}^{(k)} \eta+ L_{v}^{(k)}\mu,
\end{align}
where $ L_{\xi}^{(k)} $ and $ L_{v}^{(k)}$ are as defined in \eqref{Eq20}.
\end{thm}
\begin{proof}
The proof follows an induction-based approach. When $k=1$, the update rule \eqref{Eq6}, Assumption \ref{Assum4} and the condition \eqref{EqAssum2} show that for any $s, \overline{s} \in \mathcal{S}_{\mathrm{D}}$ and $ a\in \mathcal{A}_{\mathrm{D}}(s)$, then, 
\begin{align}\label{Eq27}
\left|\overline{q}_{\pi}^{(1)}(s, a)-\underline{q}_{\pi}^{(1)}(s, a)\right| \leqslant & |\underline{g}(s,a)-\overline{g}(s, a)|.
\end{align}
From the definitions of $\overline{g}$ and $\underline{g}$ in \eqref{Sigma_q} there exists $\xi, \overline{\xi} \in s$ and $v, \overline{v} \in a$ satisfying $ g(\xi, v)=\underline{g}(s,a) \text{ and } g(\overline{\xi}, \overline{v})=\overline{g}(s,a) $. Thus it follows from \eqref{Eq27}, the triangle inequality, and the fact that for every $\xi, \overline{\xi} \in s$, and for every $v, \overline{v} \in a$, that $\|\xi- \overline{\xi}\|\leqslant \eta, \; \|v- \overline{v}\| \leqslant \mu$ hence, 
\begin{align}\label{Eq28}
\left|\overline{q}_{\pi}^{(1)}(s, v)-\underline{q}_{\pi}^{(1)}(\overline{s}, \overline{a})\right| &\leqslant |g(\xi,v)-g(\overline{\xi}, \overline{v})| \nonumber\\
&\leqslant L_{g\xi}\|\xi-\overline{\xi}\| + L_{g v} \|v-\overline{v}\|\nonumber \\
&\leqslant L_{g\xi}\eta+L_{g v}\mu , 
\end{align} 
which satisfies \eqref{Eq23}.
Now suppose that for every $s \in \mathcal{S}_{\mathrm{D}}$, $a \in \mathcal{A}_{\mathrm{D}}(s)$,  
\begin{align}\label{Eq28,k-1}
\left|\underline{q}_{\pi}^{(k-1)}\!(s,\! a)\!-\! \overline{q}_{\pi}^{(k-1)}\!\left(s, \!a\right)\!\right| \! \leqslant \! L_{\xi}^{(k-1)}\!\eta \!+\! L_{v}^{(k-1)}\!\mu,
\end{align}
holds for all $(s,a)$ pairs. To show that \eqref{Eq23} also holds for $k$, consider the update law \eqref{Eq6}. Using the triangular inequality, one gets, 
\begin{align}\label{Eq29,k}
&\hspace{-10pt}\left|\underline{q}_{\pi}^{(k)}(s, a)-\overline{q}_{\pi}^{(k)}\left(s, a\right)\right| \leqslant \left|\underline{g}(s, a)-\overline{g}(s, a)\right| \nonumber\\
&+\gamma\left|\min\limits_{s^{\prime} \in \Delta(s,a)} \max\limits_{a^{\prime} \in \mathcal{A}_{\mathrm{D}}(s^{\prime})}\!\! \underline{q}_{\pi}^{(k-1)}(s^{\prime},a^{\prime})\right. \nonumber\\
&\;\;\;\;\;\;\;-\left. \max\limits_{s^{\prime} \in \Delta(s,a) }\max\limits_{\overline{a} \in \mathcal{A}(s^{\prime})} \overline{q}_{\pi}^{(k-1)}(s^{\prime},\overline{a}^{\prime}) \right|.
\end{align}
Using the same procedure in equation \eqref{Eq29,k} for the first term on the right hand side of  \eqref{Eq28} and use the inequality \eqref{Eq28,k-1} for the second term, yields,
\begin{align}\label{Eq29}
\hspace{-10pt}\left|\underline{q}_{\pi}^{(k)}(s, a)-\overline{q}_{\pi}^{(k)}\left(s, a\right)\right| \leqslant & L_{g\xi}\eta+L_{g v}\mu \nonumber\\
&+ \! L_{\xi}^{(k-1)}\!\eta \!+\! L_{v}^{(k-1)}\!\mu,
\end{align}
Factoring by $\eta$ and $\mu$ gives the inequality \eqref{Eq23}. This completes the proof. \hfill 
\end{proof}

\begin{rem} It is important to emphasize that:
\begin{itemize}
\item The tightness of the difference $\left|\underline{q}_{\pi}^{(k)}(s,a) - \overline{q}_{\pi}^{(k)}(s,a)\right|$ indicates the convergence of the estimates of the optimal action value function, which in turn makes the extracted policies from $\underline{q}_{\pi}^{(k)}(s,a) $ and $\overline{q}_{\pi}^{(k)}(s,a) $ similar. 
\item The result presented in Theorem \ref{THM5} provides valuable insights to select the appropriate discretization parameters $\eta$, $\mu$, the discount factor $\gamma$ and the number of episodes when applying the proposed Q-learning method to achieve a certain precision $\varepsilon$. The result highlights two cases that are worth noting:
\begin{enumerate}
\item When $L_{\xi}^{(k)} \geqslant 1$ for all $k \in \mathbb{N}$, which indicates that the system \eqref{Eq20} is forward complete (FC) \cite{pola2016symbolic}, and $L_{v}^{(k)} \leqslant C^{te}$ for all $k \in \mathbb{N}$, with $C^{te} \in \mathbb{R}_{\geqslant 0}$, the inequality \eqref{Eq23} is bounded by $\max\limits_{l \in 1,\dots, k}L_{\xi}^l \eta+L_{\xi}^l L_{v}^{k}\mu$. Hence, to maintain a desired precision (e.g., within $\varepsilon$), it is crucial to choose $\eta$ small when the number of episodes is higher.
\item When $L_{\xi}^{(k)} < 1$ and $L_{v}^{(k)} \leqslant C^{te}$ for all $k \in \mathbb{N}_{> 0}$, with $C^{te} \in \mathbb{R}_{\geqslant 0}$: In this case, inequality \eqref{Eq23} implies that there exists an $\varepsilon \in \mathbb{R}_{\geqslant 0}$ such that $\eta + C^{te}\mu \leqslant \varepsilon$. Thus, system \eqref{Eq23} exhibits the well-known incremental Input-to-State Stability ($\delta$-ISS) property, established in prior works \cite{sontag1998mathematical,khalil2002}. This property means that the bounds $\underline{q}_{\pi}$ and $\overline{q}_{\pi}$ become progressively tighter with each iteration. It implies that to achieve a desired precision ($\varepsilon$), the discretization parameter does not depend on the number of episodes.
\end{enumerate}
\end{itemize}

\hfill $\diamond$
\end{rem}

The following corollary establishes a relation between the precision between the true Q-values obtained by \eqref{update} and the minimal and maximal Q-values for the symbolic model obtained from \eqref{Eq6}.

\begin{cor}\label{Cor3}
Under Assumptions \ref{Compact}, \ref{Assum1},\ref{Assum2},\ref{Assum4}, given a number of episodes \(N\) and a desired precision \(\varepsilon>0\), if the discretization parameters \(\eta\) and \(\mu\) are chosen such that,
\begin{align}\label{Eq34}
L_{\xi}^{\max}\eta + L_{v}^{\max}\mu \leqslant \varepsilon,
\end{align}
where 
$L_{v}^{max}\geqslant L_{v}^{(k)}$ and $L_{\xi}^{max}\geqslant L_{\xi}^{(k)}$ for all $\xi\in s, v\in a$, \(k = 1, 2, \dots, N\), it follows that, 
\begin{subequations}\label{Eq35}
\begin{gather}
\left|\underline{q}_{\pi}^{(k)}(s,a)-\overline{q}_{\pi}^{(k)}(s,a)\right| \leqslant \varepsilon,\\ \label{Eq26.2}
\left|q_{\pi}^{(k)}(\xi,v)-\overline{q}_{\pi}^{(k)}(s,a)\right| \leqslant \varepsilon, \\ \label{Eq26.3}
\left|\underline{q}_{\pi}^{(k)}(s,a)-q_{\pi}^{(k)}(\xi,v)\right| \leqslant \varepsilon.
\end{gather}
\end{subequations}
\end{cor}

\begin{pf}
The proof is straightforward from the use of \eqref{Eq34} as an upper bound of the right hand side of \eqref{Eq23}. \hfill
\end{pf}

\begin{rem}
Conditions \eqref{Eq35} express how  far are the estimated  Q-values of discrete symbolic model \eqref{Eq6} from the true Q-values of continuous state action spaces, e.g., \textcolor{black}{if $\left|\underline{q}_{\pi}^{(k)}(s,a) - \overline{q}_{\pi}^{(k)}(s,a)\right|=0$ the value loss is null and the extracted policy is optimal}.  Given a number of episodes $N$ and precision $\varepsilon>0$, appropriately choosing discretization parameters $\eta$ and $\mu$ ensures an $\varepsilon$-approximation of true Q-values. 
\end{rem}
Figure \ref{fig:optimalq} provides a visual representation of the key results derived from Theorem \ref{THM3} and Corollary \ref{Cor3}. In Theorem \ref{THM3}, it is proven that the continuous Q-values are bounded by the lower and upper bounds, $\underline{q}_{\pi}^{(k)}$ and $\overline{q}_{\pi}^{(k)}$, respectively. This demonstrates the range within which the Q-values can vary. Corollary \ref{Cor3} complements the theorem by establishing an important property. It shows that the maximum difference between the lower and upper bounds is bound by the value $\varepsilon$. This bound is determined by the dependencies on $\mu$ and $\eta$. Further analysis focuses on understanding the behavior of the quantity $\left|\underline{q}_{\pi}^{(k)}(s,a)-\overline{q}_{\pi}^{(k)}(s,a)\right|$ and its asymptotic stability. The analysis explores how this quantity approaches a stable value over time. Such stability analysis can be performed on the basis of the stability of the difference equation \eqref{Eq20}, using the well-established result on the stability of linear difference systems by \cite{LaSalle1986}. The following theorem provides a direct result on the choice of the discount factor, and the reward gains associated to equation \eqref{Eq20}, making it possible to learn Q values and achieve a certain accuracy \(\varepsilon\) that tend to $ L_{g\xi}\eta+L_{gv}\mu$.


\begin{thm}
Consider the difference equation \eqref{Eq20}, given a discount factor $\gamma \in \left(0,1\right)$ if there exists positive definite matrix $P$ with appropriate dimension such that, 
\begin{align}\label{Eq37}
\gamma^{2}\left(\begin{array}{cc}
L_{f\xi}     &    L_{\mathcal A} \\
L_{f\xi}     &   L_{\mathcal A}
\end{array}\right)^{T}P\left(\begin{array}{cc}
L_{f\xi}     &    L_{\mathcal A} \\
L_{f\xi}     &   L_{\mathcal A}
\end{array}\right)-P < 0
\end{align}
then, equation \eqref{Eq20} is asymptotically stable and the precision $\varepsilon $ tend to the equilibrium $L_{0}= \left[\begin{array}{cc}
L_{g\xi}\\L_{gv}
\end{array}\right]$ as long as the number of episodes increases.
\end{thm}
\begin{pf}
Consider the following Lyapunov function $V(k)=\left[\begin{array}{c}
L_{\xi}^{(k)} \\ L_{v}^{(k)}
\end{array}\right]^{T}P\left[\begin{array}{c}
L_{\xi}^{(k)} \\ L_{v}^{(k)}
\end{array}\right]$. Note that $V(L_{\xi}=0, L_{v}=0)=0$ and because $P$ is positive definite, the function $ V $ is positive. Now, take its forward difference along the trajectory of system \eqref{Eq20} around the zero equilibrium, it follows that,
$$\begin{aligned}
V(k+1)-V(k)=&\left[\begin{array}{c}
L_{\xi}^{(k)} \\ L_{v}^{(k)}
\end{array}\right]^{T}\!\!\left(\gamma^{2}\left(\begin{array}{cc}
L_{f\xi}     &    L_{\mathcal A}\\
L_{f\xi}     &   L_{\mathcal A}
\end{array}\right)^{T}P\right. \nonumber\\
&\left. \times\left(\begin{array}{cc}
L_{f\xi}     &    L_{\mathcal A} \\
L_{f\xi}     &   L_{\mathcal A}
\end{array}\right)-P\right)\left[\begin{array}{c}
L_{\xi}^{(k)} \\ L_{v}^{(k)}
\end{array}\right].
\end{aligned} $$
If condition \eqref{Eq37} is satisfied, then $\Delta V<0$, which implies asymptotic stability of the system \eqref{Eq20}. This completes the proof. \hfill$\square$
\end{pf}
\begin{rem}
Given an arbitrary discount factor $\gamma\in \left(0,1\right) $, inequality \eqref{Eq37} becomes a linear matrix inequality and thus it can be solved using standard semidefinite programming (SDP) solvers, \cite{boyd1994linear}.
\end{rem}

\section{Controller refinement and Q-learning for symbolic model}	
Q-learning is an off-policy reinforcement learning algorithm, offering the advantage of evaluating and comparing maximally and minimally optimal policies in an offline manner. This flexibility allows the designer to select the policy that best suits the control problem's requirements.

The proposed Q-learning algorithm is shown in Algorithm \ref{alg:q_learning_symbolic}. Once the learning is established, any of the Q-tables $\underline{q}$ and $\overline{q}$ can be used to refine the learned strategy into a controller for the original system $\Sigma$ in (\ref{Eq1}). Formally, a controller for the original discrete-time system $\Sigma$ in (\ref{Eq1}) is a map $\pi:\mathcal{S} \rightarrow \mathcal{A}$ satisfying $\pi(\xi) \in \mathcal{A}(\xi)$ for all $\xi \in \mathcal{S}$. Now given the results of the $Q$ tables $\underline{q}$ and $\overline{q}$, two control strategies are generated:
\begin{itemize}
    \item A lower control strategy $\underline{\pi}:\mathcal{S} \rightarrow \mathcal{A}$ defined formally as follows: for $\xi \in \mathcal{S}$ such that $\xi \in s$, $s\in \mathcal{S}_{\mathrm{D}}$, we have $\underline{\pi}(\xi)\in\arg \max\limits_{a} \underline{q}(s,a)$;
    \item  An upper control strategy $\overline{\pi}:\mathcal{S} \rightarrow \mathcal{A}$ defined formally as follows: for $\xi \in \mathcal{S}$ such that $\xi \in s$, $s\in \mathcal{S}_D$, we have $\overline{\pi}(\xi)\in\arg \max\limits_{a} \overline{q}(s,a)$;
\end{itemize}

One may be in the case where there is a disagreement between the policies generated by $\underline{q}$ and $\overline{q}$, if they disagree, one can refine either using $\underline{q}$ or $\overline{q}$ while preserving $\varepsilon$-approximate optimality, ensuring (\ref{Eq35}). Remark 5 in the Appendix 2 highlight other differences between algorithm \ref{alg:q_learning_symbolic} and the algorithm in \cite{borri2023reinforcement}.
\setcounter{algorithm}{2}
\begin{algorithm}
	\caption{Q-Learning from symbolic models}
	\label{alg:q_learning_symbolic}
	\begin{algorithmic}[1]
		\STATE \textbf{Input:} Environment with continuous state space $\mathcal{S}$, continuous action space $\mathcal{A}$. Define the Lipshitz constant $L_{f}$ and $L_{g}$, the learning rate $\alpha$ and the discount factor $\gamma$. Choose an arbitrary Q-value precision $\varepsilon$. Initialize $\underline{q}$ and $\overline{q}$ satisfying Assumption \ref{Assum4}.
		\STATE \textbf{Output:} Learned $\underline{q}$ and $\overline{q}$-values for all symbolic states and action samples.
		\STATE Compute the discretisation parameters $\eta$, $\mu$ based on $\varepsilon$ and \eqref{Eq35}.
            \STATE Derive the symbolic model $\Sigma_{\mathrm{D}}$.
		\FOR{each episode}
		\STATE Initialize state $s$
		\WHILE{episode not finished}
		\STATE Choose action $a \in \mathcal{A}_{q}$ using an exploration or exploitation strategy (e.g., epsilon-greedy)
		\STATE Execute action $a$ and observe all possible successors $s^{\prime} \in \Delta(s,a)$ and rewards $\underline{g}(s,a)$,  $\overline{g}(s,a)$.
		\STATE Update the maximally and minimally Q-values using the Q-learning update rule:\\
		\STATE $\underline{q}(s, a) \leftarrow \underline{q}(s, a) + \alpha(\underline{g} + \gamma\min\limits_{s^{\prime} \in \Delta(s,a)}\max\limits_{a^{\prime}\in \mathcal{A}_{q}(s^{\prime})} \underline{q}(s', a'))$
		\STATE $\overline{q}(s, a) \leftarrow \overline{q}(s, a) + \alpha(\underline{g} + \gamma\max\limits_{s^{\prime} \in \Delta(s,a)}\max\limits_{a'\in \mathcal{A}_{q}(s')} \overline{q}(s', a'))$
		\STATE Update current state using any state that belongs to  $\Delta(s,a)$
		\ENDWHILE
		\ENDFOR
	\end{algorithmic}
\end{algorithm}

\section{Case study}
We test our proposed method in two different control
problems, the monttain car example and the van der pol oscillator.
\subsection{Mountain Car Control problem}
Consider $\Sigma$ as the discrete-time nonlinear system of the Mountain Car as described in \cite{Moore90efficientmemory-based}:
\begin{align}\label{DynamicMontainCar}
\xi_{k+1} =A\xi_{k}+ \Phi(\xi_{k}) + B v_{k}  
\end{align}
where the matrices $A, B$, and the non-linearity $\Phi_{k}(\xi)$ are given by, $$
\begin{aligned}
&A= \left(\begin{array}{cc}
1 & 1\\
0 & 1
\end{array}\right), \; B=\left(\begin{array}{c}
0\\
0.001
\end{array}\right)\\
&\Phi(\xi_{k})=\left(\begin{array}{c}
0\\
- 0.0025 \cos(3\xi_{1}(k))
\end{array}\right)
\end{aligned}$$
The Mountain Car state, $\xi(k)= \left(\begin{array}{cc}
\xi_{1}^{T}(k) & \xi_{2}^{T} (k)
\end{array}\right)^{T}$, is defined by $\xi_1 \in [\xi_{1,\min}, \xi_{1,\max}]$ the position of the car, and $\xi_2 \in [\xi_{2,\min}, \xi_{2,\max}]$ its velocity.  The specific values for the minimum and maximum ranges are, $\xi_{1,\min} = -1.2$;  $\xi_{1,\max} = 0.6$; $\xi_{2,\min} = -0.07$; $\xi_{2,\max} = 0.07$.  The control input is denoted by $v \in \mathcal{A} = [-1, 1]$, where $\mathcal{A}$ represents the continuous action space. 

The Mountain Car problem is a deterministic difference equation that involves a car placed at the bottom of a sinusoidal valley. The car has the ability to apply accelerations $v_{k}$ in the set $\mathcal{A}$. The objective is to strategically accelerate the car in order to reach the goal state situated on top of the hill defined by the position $\xi_{1}^{\ast}=\xi_{1,\max}=0.6$. To achieve this goal, the model in gym environment, \cite{Moore90efficientmemory-based}, assigns a reward of $g(\xi,v)=0$ to the state $\xi_{1}^{\ast}$, indicating a successful achievement of the goal state and a reward of $g(\xi,v)=-1 $ is assigned to all other states in the environment, representing unsuccessful or intermediate states. 

As the state and action spaces are in a continuous domain, it is impossible to apply the standard Q-learning approach directly. Following the proposed space discretization method, define $\eta_{\xi}$ as the number of discrete states for each $\xi_{i}$, and consider $\eta_{i}=\frac{\xi_{i,max}-\xi_{i,min}}{\eta_{\xi}}$, \(i=1,2\)  and $ \mu=1$, a symbolic model conforming to \eqref{Sigma_q} is constructed. The construction of a symbolic model of $\Sigma$ involves the computation of the over-approximation of attainable sets for $s\in \mathcal{S}_{\mathrm{D}}$. Thus, to compute the local Lipschitz constant over a discrete state, take two states $\xi, \overline{\xi} \in s $ under different control inputs $v_{k},  \in \mathcal{A}(\xi), $ $ \overline{v}_{k} \in \mathcal{A}(\overline{\xi})$, then, it yields,
\begin{align*}
\left|\xi_{k+1}-\overline{\xi}_{k+1}\right| =& \resizebox{0.6\hsize}{!}{$\left|A (\xi_{k}-\overline{\xi}_{k})+\left(\begin{array}{c}
0\\
-0.0025\left(cos(3\xi_{1}(k))-cos(3\overline{\xi}_{1}(k))\right)
\end{array}\right) + B (v-\overline{v})\right|$}\\
\leqslant& \resizebox{0.6\hsize}{!}{$\left|A (\xi_{k}-\overline{\xi}_{k})+\left(\begin{array}{c}
0\\
-0.0025
\end{array}\right)\right|+B\left|(v_{k}-\overline{v}_{k})\right|.$}
\end{align*}
Hence, the local Lipschitz constants for system \eqref{DynamicMontainCar} are over-approximated by, $L_{f \xi}= \|A\|+0.0025=1.0025$ and $L_{f v } = \|B\|=0.001$. 
Now, let $s^{\ast} \in \mathcal{S}_{\mathrm{D}}$ denotes the discrete state to which the original goal position belongs, which is the position of the car at the top of the hill. To represent the goal state, the maximal and minimal reward functions are assigned as follows: $\overline{g}(s^{\ast})=0$,  $\underline{g}(s^{\ast})=0$ and a reward of $\underline{g}(s)=\overline{g}(s)=-1, \; \forall s \in \mathcal{S}_{\mathrm{D}}\backslash\left\lbrace s^{\ast} \right\rbrace$. To demonstrate the effectiveness of the proposed approach, three experiments are conducted. The first experiment aims to showcase the response of the obtained policies using the Q-learning based symbolic controller. The second experiment focuses on quantifying the achieved precision under different discretization parameters. Lastly, the third experiment investigates the convergence of the maximally and minimally optimal policies towards a unique optimal policy as the discretization parameters become finer. Detailed descriptions of each experiment are provided below.

\paragraph{Experiment 1}
For $n_{\xi} = 160, \;\alpha=0.4, \; \gamma=0.99 $ and an exploration probability $0.4$, Algorithm \ref{alg:q_learning_symbolic} is utilized to compute the maximally and minimally optimal policies for the mountain car example. The optimal maximally and minimally policies are plotted in Figures \ref{fig:ExtractedOptimalPolicyQmin}-\ref{fig:ExtractedOptimalPolicyQmax}. According to Section (controller refinement), two controllers were refined from both $\underline{q}_{\pi}^{\ast}$ and $\overline{q}_{\pi}^{\ast}$. The trajectory of the car under each policy is depicted in Figures \ref{fig:CarTrajectoryMin}-\ref{fig:CarTrajectoryMax}. This experiment demonstrates that both the maximally and minimally optimal policies enable the car to successfully reach the control goal, which is to balance the car's acceleration and reach the top of the hill.

\begin{figure}
\centering
\includegraphics[width=1\linewidth]{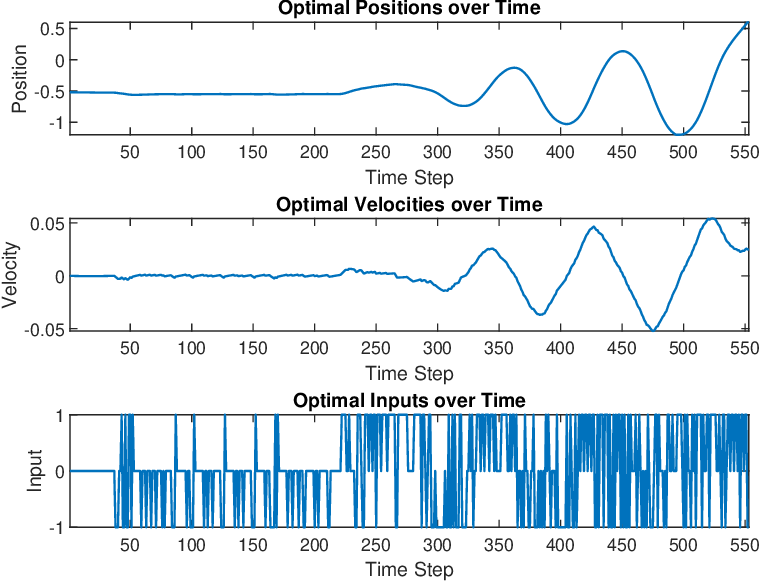}
\caption{Car trajectory obtained using the refined optimal policy derived from the minimal Q-values using Algorithm \ref{alg:q_learning_symbolic}.}
\label{fig:CarTrajectoryMin}
\end{figure}
\begin{figure}
\centering
\includegraphics[width=1\linewidth]{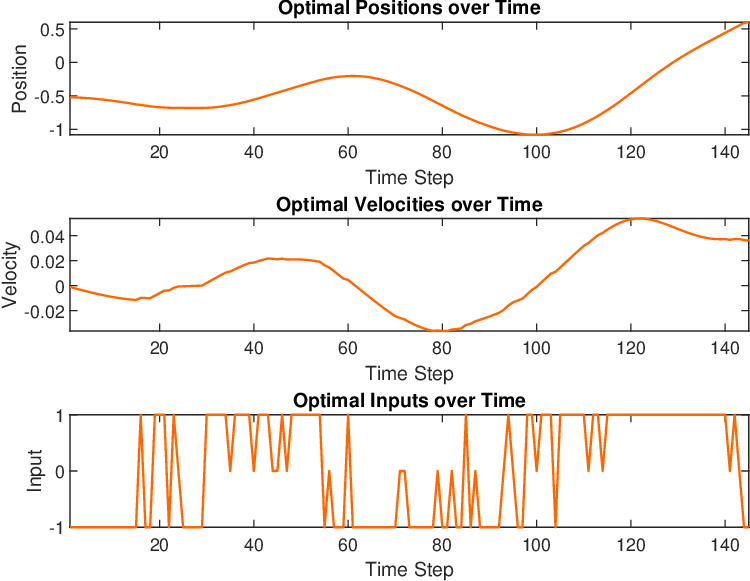}
\caption{Car trajectory obtained using the refined optimal policy derived from the maximal Q-values using Algorithm \ref{alg:q_learning_symbolic}.}
\label{fig:CarTrajectoryMax}
\end{figure}


\begin{figure}
\centering
\includegraphics[width=1\linewidth]{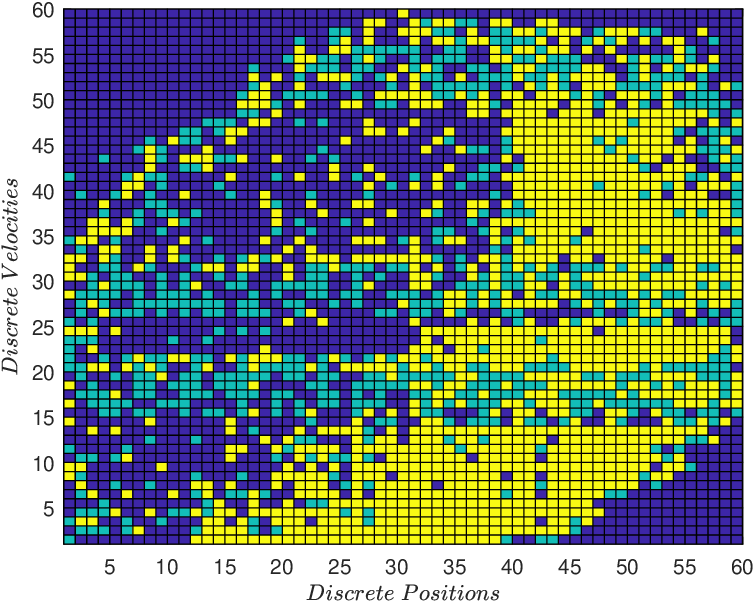}
\caption{Extracted optimal policy $\overline{q}_{\pi}$.}
\label{fig:ExtractedOptimalPolicyQmin}
\end{figure}
\begin{figure}
\centering
\includegraphics[width=1\linewidth]{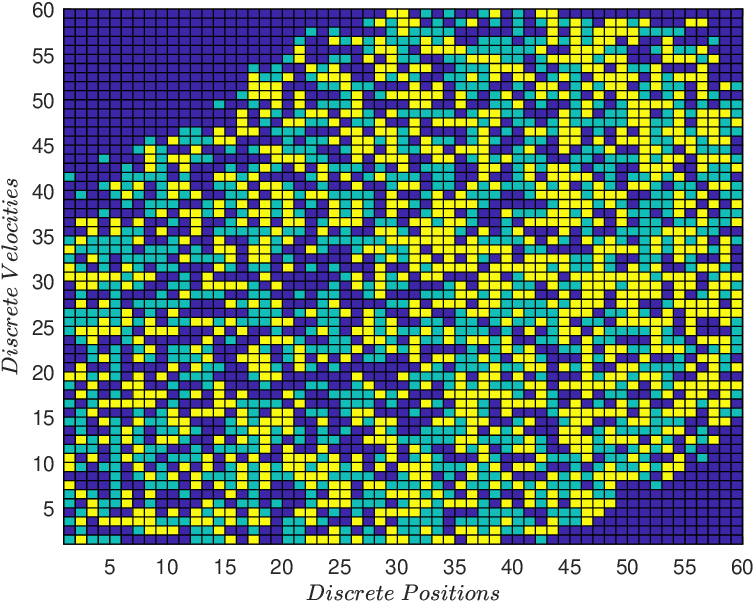}
\caption{Extracted optimal policy $\underline{q}_{\pi}$.}\label{fig:ExtractedOptimalPolicyQmax}
\end{figure}

\paragraph{Experiment 2}
The following experiments compute the optimal policies for different values of $n_{\xi}$, namely $n_{\xi}=40, 80, 160$ and different symbolic inputs $n_{v}=3, 6, 12$. The learning parameters $\alpha=0.4$, $\gamma=0.5$, and an exploration probability of $0.4$ are used. Table \ref{tab:1} shows the measured distance $\max\limits_{s\in \mathcal{S}_{\mathrm{D}}} \max\limits_{a\in \mathcal{A}_{\mathrm{D}}}\left|\underline{q}_{\pi}(s,a)-\overline{q}(s,a)\right|$ experimentally. The primary outcome is that the distance $\max\limits_{s\in \mathcal{S}_{\mathrm{D}}} \max\limits_{a\in \mathcal{A}_{\mathrm{D}}}\left|\underline{q}_{\pi}(s,a)-\overline{q}(s,a)\right|$ is bounded, confirming the theoretical result given in Theorem \ref{THM5}. The second observation is that by reducing the discretization parameters, we can achieve tighter bounds, leading to minimal and maximal Q-values that closely resemble those of continuous state-action spaces. This indicates that the refined Q-values exhibit a higher level of accuracy, approximating the Q-values in the continuous setting.  


\begin{table}
\caption{Precision ($\varepsilon$) for Different Space Discretization: $\eta_{\xi}$ (number of discrete states) and $n_{v}$ (number of symbolic inputs).}
\label{tab:1}
\centering
\begin{tabular}{|c|c|c|c|}
\hline
\diagbox[width=3cm, height=2cm]{\\ $n_{v}$}{ \\ \hspace{8pt} $n_{\xi}$ } & 40 & 80 & 160 \\
\hline
3 & 2 & 1.11 & 1.27 \\
\hline
6 & 1.99 & 1.8295 & 1.1424 \\
\hline
12 & 1.98 & 1.18 & 0.8558\\
\hline
\end{tabular}
\end{table}

\paragraph{Experiment 3}
This simulation compares the similarity between the maximally and minimally optimal policies obtained for different discretization parameters $n_{\xi}=40, 60, 80, 100, 120, 140, 160$ and $n_{v}=3$. The non-similarity ratio, denoted as $\rho$, is defined as the number of non-similar state-action pairs divided by $\eta_{\xi}^{2}$, i.e., $$ \rho= \frac{\text{Number of Nonsimilar state action pairs}}{\eta_{\xi}^{2}}. $$
A "Nonsimilar state-action pair" is formally defined by $\overline{a}(s)\neq \underline{a}(s)$ when considering a state $s \in \mathcal{S}_{\mathcal{D}}$. Here, $\underline{a}(s), \overline{a}(s)$ represents the action to be performed in state $s$ using the minimal and maximal policies, respectively. This definition captures situations where  $\overline{a}(s)$ and $\underline{a}(s)$ for a given state have different action values, indicating the potential for different actions to be chosen. In fact, a Nonsimilar state-action pair highlights cases where the action values  $\overline{a}(s)$ and $\underline{a}(s)$ suggest different optimal actions for a given state. This distinction becomes important in understanding the variability and possible divergence in action choices under both policies.
Figures \ref{fig:ExtractedOptimalPolicyQmax} and \ref{fig:ExtractedOptimalPolicyQmin} illustrate the refined optimal policies for the maximal and minimal Q-values, respectively. In these figures, the blue color represents the action $"-1"$, the blue sky color represents the action $"0"$, and the yellow color represents the action $"1"$. Remarkably, both optimal policies exhibit similar patterns, capturing similar state-action pairs despite their differences in the maximum and minimum Q-values. This means that the refined optimal policies retain essential characteristics from the continuous state-action space, with tighter bounds achieved by reducing the discretization parameters.
Figure \ref{fig:Tightness} illustrates the relationship between the ratio $\rho$ and the level of discretization. It can be observed that as the discretization becomes finer, the value of $\rho$ decreases. This indicates that as the intervals in the state space become smaller, the bounds $\left|\underline{q}_{\pi}^{\ast}-\overline{q}_{\pi}^{\ast}\right|$ become tighter, indicating an increased level of precision in estimating the optimal policy of the continuous state-action space.
\begin{figure}
\centering
\includegraphics[width=1\linewidth]{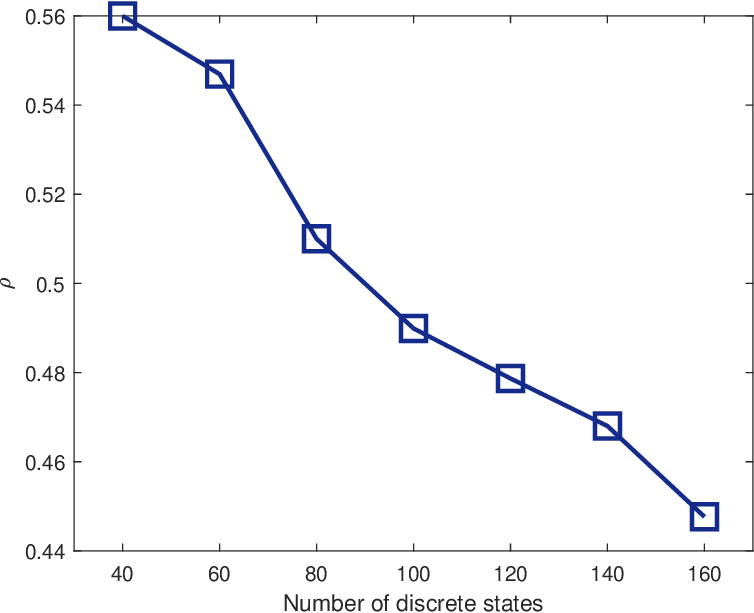}
\caption{Similarity between the extracted optimal policies vs Tightness of $ \left|\underline{q}_{\pi}^{\ast}-\overline{q}_{\pi}^{\ast}\right|$ for different values of $\eta_{\xi}$.}
\label{fig:Tightness}
\end{figure}

\subsection{Van der Pol Oscillator}
{Consider the following variant of the Van Der Pol Oscillator, which incorporates an external control input \(u_{k}\)  acting on the velocity \(v\):
\vspace{-8pt}
	\begin{align}\label{Ar1}
		x_{k+1} &= x_{k} + \tau v_{k}\\
		v_{k+1} &= v_{k} + \tau (\zeta \cdot (1 - x_{k}^2) \cdot v - x + u_{k})
	\end{align}
	Here, \(x\) denotes the position, and we have the following setup parameters: \(\tau = 0.01\), \(x_k \in [-2, 2]\) (range of \(x_k\) for the position), \(v_k \in [-3, 3]\) (range of \(v_k\) for the velocity), and \(u_k \in \{-1,0,1\}\) (action space). The parameter $\zeta$ is set $\zeta=2$. The desired position and velocity, denoted as \(x_{desired}\) and \(v_{desired}\), are both set to zero, indicating that the goal is to stabilize the system at the origin. To achieve this objective, a reward function is constructed \(g(x, v) = -((x - x_{desired})^2 + (v - v_{desired})^2)\), which penalizes the deviation of the system's position and velocity from the desired state. Here, once the run reaches the state representing the origin the simulation is stopped. To apply standard Q-learning \cite{watkins1989learning}, a space discretization method is required for both the position and velocity spaces. Apply the proposed abstraction with $\eta=\frac{1}{20}$ and $\mu = \frac{1}{20}$. Choosing the learning rate $\alpha=0.5$ and a discount factor $\gamma=0.9$, Algorithm \ref{alg:q_learning_symbolic} is performed. }
 Figure \ref{fig:x2} compares the trajectories of the system using the Q-learning based uniform discretisation in Algorithm \ref{alg:q_learning_discretisation} and the proposed Q-learning based symbolic abstraction. Within the above simulation setup, the trajectory of the system with uniform discretisation exhibits oscillatory behavior and does not converge or passes through the origin. On the other hand show the controlled closed-loop system using the proposed approach. Indeed, it is shown that the trajectory enters the red circle, indicating that the associated controller has achieved the control goal.
Q learning in Algorithm 

\begin{figure}[!h]
\vspace{-10pt}
    \begin{minipage}[t]{0.48\linewidth}
        \includegraphics[width=1\linewidth]{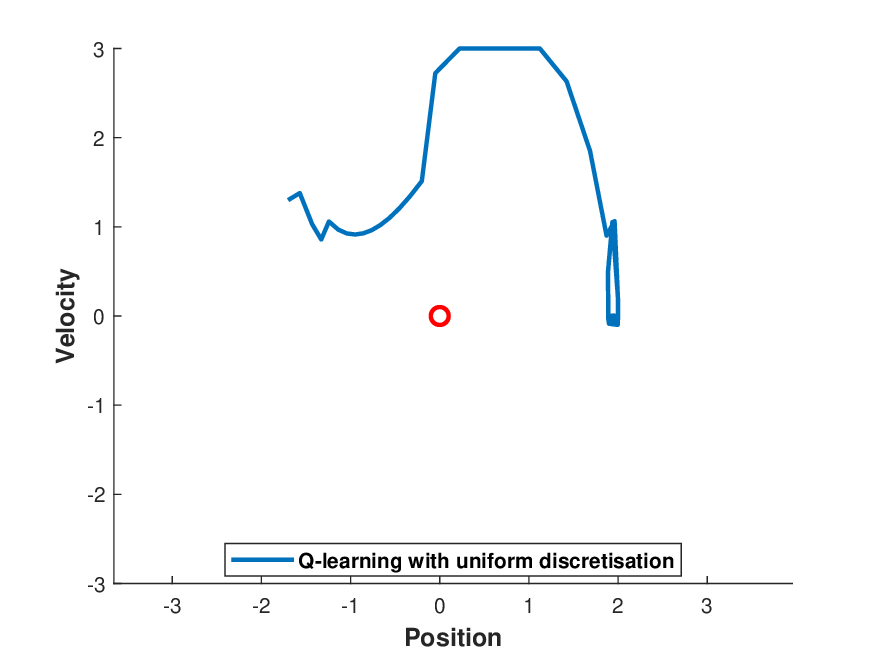}
        \subcaption{System trajectory with uniform discretisation.}
        \label{Fig:OpenLoop}
    \end{minipage}
    \hspace{0.1cm}\begin{minipage}[t]{0.48\linewidth} 
        \includegraphics[width=1\linewidth]{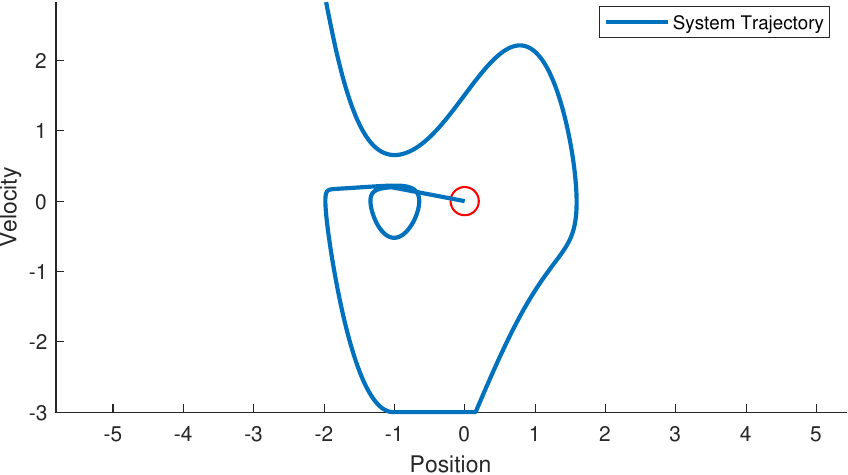}
        \subcaption{System trajectory with Q-learning based symbolic abstraction.}
        \label{Fig:ClosedLoop}
    \end{minipage}      
    \caption{Trajectories Van Der Pol oscillator using Algorithms \ref{alg:q_learning_discretisation} and \ref{alg:q_learning_symbolic}.}
    \label{fig:x2}  
    
\vspace{-15pt}
\end{figure}

\section{Conclusion}
This work has focused on addressing the challenge of applying Q-learning to control problems with continuous state and action spaces. Based on the fact that space discretisation methods, i.e., uniform and Voronoi are not able to capture systems trajectories due to the reachable induced mismatch. This work proposes symbolic discretization-based approach, which enables to capture systems trajectories through the framework of alternating simulation relation from the abstraction to the concrete system. Such abstraction takes the form of  a non-deterministic finite transition systems, which is not featured with transition probabilities. By extending and generalizing previous work in \cite{powell2007approximate,sutton2018reinforcement,borri2023reinforcement},Two Q-tables, namely the minimal and maximal Q-tables, are learned which effectively bound the Q-values of the continuous state-action space. Theoretical insights have been provided, establishing a relation between the tightness of the Q-value bounds and the abstraction parameters. Analysis has shown that reducing the distance between quantizer levels leads to tighter Q-values, ultimately resulting in a convergence of the extracted policy towards the optimal policy.	Moreover, through the evaluation of our algorithm on the mountain car control problem, the convergence of the Q-learning process and the achieved precision in approximating the optimal Q-values is demonstrated. The proposed approach opens up new possibilities for addressing reachable-induced mismatch and improving the performance of reinforcement learning algorithms in continuous domains. Future research can further explore the applicability of this approach to the class of nonlinear systems continuous time and state action spaces.

\bibliographystyle{ieeetr}
\bibliography{IEEEbib}

\end{document}